\documentclass[a4paper,10pt]{article}
\usepackage[utf8]{inputenc}
\usepackage{geometry} 
\usepackage{graphicx}
\usepackage[font={small},width=0.55\textwidth]{caption}
\usepackage{lipsum}
\usepackage[usenames, dvipsnames]{color}
\usepackage{fancyhdr}
\usepackage{mathtools}
\usepackage{float}
\usepackage{blindtext}
\usepackage{mathptmx}
\usepackage{enumitem}
\usepackage{amsthm}
\usepackage{dsfont}
\usepackage{eucal}
\usepackage{stfloats}
\usepackage{titlesec}
\usepackage{adjustbox}
\usepackage{microtype}

\titleformat{\section}
{\normalfont\Large}{\thesection}{1em}{}
\titleformat{\subsection}
{\normalfont\large}{\thesubsection}{1em}{}

\usepackage[ruled,vlined]{algorithm2e}
\usepackage{authblk}
\usepackage{amsmath,amsfonts,amssymb}

\usepackage[backend = biber, style=numeric-comp,isbn=false,url=false,doi=false,eprint=false,sorting=none]{biblatex}
\renewbibmacro{in:}{}

\DeclareCiteCommand{\tabcite}
  {\usebibmacro{cite:init}%
   \usebibmacro{prenote}}
  {\usebibmacro{citeindex}%
   \usebibmacro{cite:comp}}
  {}
  {\usebibmacro{cite:dump}%
   \usebibmacro{postnote}}

\usepackage{array}
\newcolumntype{P}[1]{>{\centering\arraybackslash}p{#1}}

\newtheorem{theorem}{Theorem}
\newtheorem{lemma}[theorem]{Lemma}

\newcommand{\ket}[1]{\ensuremath{|#1\mkern-1mu\rangle}}
\newcommand{\dyad}[1]{\ensuremath{|#1\rangle\mkern-3mu\langle #1|}}
\newcommand{\ddyad}[2]{\ensuremath{|#1\rangle\mkern-3mu\langle #2|}}
\newcommand{\braket}[1]{\ensuremath{#1}}

\DeclarePairedDelimiter\ceil{\lceil}{\rceil}

\usepackage{xspace}
\def\cz{\ensuremath{\text{CZ}}\xspace}
\def\czlo{\ensuremath{\text{CZ}^{\mathrm{LO}}}\xspace}
\def\rpeg{\ensuremath{\text{R-PEG}}\xspace}

\newcommand{\new}[1]{\textcolor{Black}{#1}}

\title{Hard limits on the postselectability of optical graph states}

\author[1]{Jeremy C. Adcock}
\author[1]{Sam Morley-Short}
\author[1]{Joshua W. Silverstone\thanks{josh.silverstone@bristol.ac.uk}}
\author[1]{Mark G.  Thompson}
\affil[1]{QETLabs, H. H. Wills Physics Laboratory \& Department of Electrical and Electronic Engineering, University of Bristol, Merchant Venturers Building, Woodland Road, Bristol BS8 1UB, UK}

\date{June 2018}

\begin{document}
\maketitle
\thispagestyle{fancy}

\textbf{
Coherent control of large entangled graph states enables a wide variety of quantum information processing tasks, including error-corrected quantum computation. 
The linear optical approach offers excellent control and coherence, but today most photon sources and entangling gates---required for the construction of large graph states---are probabilistic and rely on postselection.
In this work, we provide proofs and heuristics to aid experimental design using postselection.
\new{We introduce a versatile design rule for postselectable experiments: drawn as a graph, with qubit modes as vertices and gates and photon-pair sources as edges, an experiment may only contain cycles with an \emph{odd} number of sources.}
We analyse experiments that use photons from postselected photon-pair sources, and lower bound the number of accessible classes of graph state entanglement in the non-degenerate case---graph state entanglement classes that contain a tree are are always accessible.
The proportion of graph states accessible by postselection shrinks rapidly, however. We list accessible classes for various resource states up to 9 qubits.
Finally, we apply these methods to near-term multi-photon experiments.
}

\section{Introduction}

\new{Postselection---whereby the success of a probabilistic photonic process is confirmed only after all the photons have been measured---}has been a testbed for fundamental quantum phenomena since its inception \cite{freedman1972experimental, sychev2017enlargement, peruzzo2012quantum, kwiat1995experimental, aspect1982experimental,bouwmeester1999observation} and although it has been shown that large-scale linear-optical quantum computing is possible in principle, it requires mid-computation measurement and feed-forward \cite{knill2001scheme, gimeno2015three, morley2017physical, pant2017percolation}. Modern schemes rely on the generation of large entangled states, on which measurement-based quantum computation is performed \cite{raussendorf2001one}. Integrated quantum photonics \cite{silverstone2016silicon, crespi2016suppression,  wang2018multidimensional} is one exciting route to large-scale applications, but it requires the on-chip generation of many-qubit quantum states.

\new{Graph states illuminate the entanglement in measurement-based protocols \cite{hein2006entanglement}, finding wide usage throughout  quantum information: from measurement-based quantum computing \cite{raussendorf2001one}, to error correcting codes \cite{schlingemann2001quantum} and quantum communications \cite{markham2008graph}. Different graph states are suited to different applications, so the ability to generate a variety of graph states is powerful. In photonic systems, where quantum gates are normally performed on fixed hardware, measurement-based protocols (based on graph states) use only the choice of measurement basis to provide true programmability.}

Photons are notoriously difficult to both produce and interact on-demand. This has led to slow improvements in photon number \cite{bouwmeester1999observation, pan2001experimental, walther2005experimental, lu2007experimental, yao2012observation, wang2016experimental}. Today, the most common way to produce quantum states of light is via what we will refer to as entangled postselected pair (EPP) sources, such those based on parametric down-conversion or spontaneous four-wave mixing. These processes can produce pair-wise entanglement and have been ubiquitous in photonic quantum information experiments over the last thirty years \cite{shalm2015strong, giustina2015significant, hong1987measurement, aspect1982experimental}. 

Large entangled states remain a challenge, however. Two common ways to generate entanglement in linear optics are: the postselected controlled-$Z$ (\cz) gate \cite{hofmann2002quantum, ralph2002linear}; and the postselected fusion gate \cite{browne2005resource}. So far, up to ten photons have been entangled in this way \cite{wang2016experimental}, though ambitious large-scale proposals exist \cite{bodiya2006scalable, lin2011weaving}. References \cite{krenn2017quantum} and \cite{gu2018quantum} use graph theory to determine which high-dimensional Greenberger-Horne-Zeilinger (GHZ) state is produced by any configuration of EPP sources.

Probabilistic gates are provably the only way to generate entanglement using linear optics \cite{knill2001scheme}, and are a central component of modern linear-optical quantum computing schemes \cite{gimeno2015three, morley2017physical, li2015resource, pant2017percolation}. Currently, postselected entangling gates (PEGs) are the only way to test linear optical devices and techniques in the multi-photon regime. The complexity of generating graph-states using linear optics has not yet been analysed \cite{mhalla2004complexity, cabello2011optimal}. 

Here, we analyse connected graph states of ``dual-rail'' photonic qubits and show that postselected gates, as well as having exponential time complexity, are fundamentally limited in which types of entanglement they can produce. This limitation, combined with the result of ref. \cite{knill2001scheme}, signals the end of passive linear optics as the universal testbed for quantum phenomena---\new{in the near future, most quantum states will not be accessible by postselection alone.}

\new{This paper is organised as follows. In section \ref{sec:gsfromlo} we review the graph state formalism, and the effects of linear optical operations on graph states. In section \ref{sec:pslimits}, we derive methods to establish the postselectability a given experimental configuration. In section \ref{sec:whichps}, we numerically explore the space of postselectable graph-state preparation experiments and determine which graphs can be reached. We summarise our findings in section \ref{sec:conclusion}}


\section{Graph states from linear optics}
\label{sec:gsfromlo}

We first introduce graph states and the special graph operation, local complementation.

Graph states are $n$-qubit stabiliser states which have a direct correspondence to undirected $n$-vertex (order $n$) graphs. A graph state, $\ket{G}$, corresponding to the graph $G=(V,E)$ with vertices $V$, and edges $E$, is written:
	\begin{equation}
	\ket{G} = \prod_{(i,j) \in E} \cz_{ij} \ket{\text{+}}^{\otimes |V|}
	\end{equation}
Where $\ket{+} = (\ket{0} + \ket{1})/\sqrt{2}$ and $ \cz = \dyad{00} + \dyad{01} + \dyad{10} - \dyad{11} $. Graph states are thus real, equal-weight states. 

All stabiliser states can be transformed into some graph state using local operations \cite{anders2006fast}. It is well known, for example, that star-type graph states are locally equivalent to GHZ states. Most states, though, are not locally equivalent to a graph.

\begin{figure*}[t!]
\centering

\captionsetup{width=0.90\textwidth}
\includegraphics[width=1.0\textwidth]{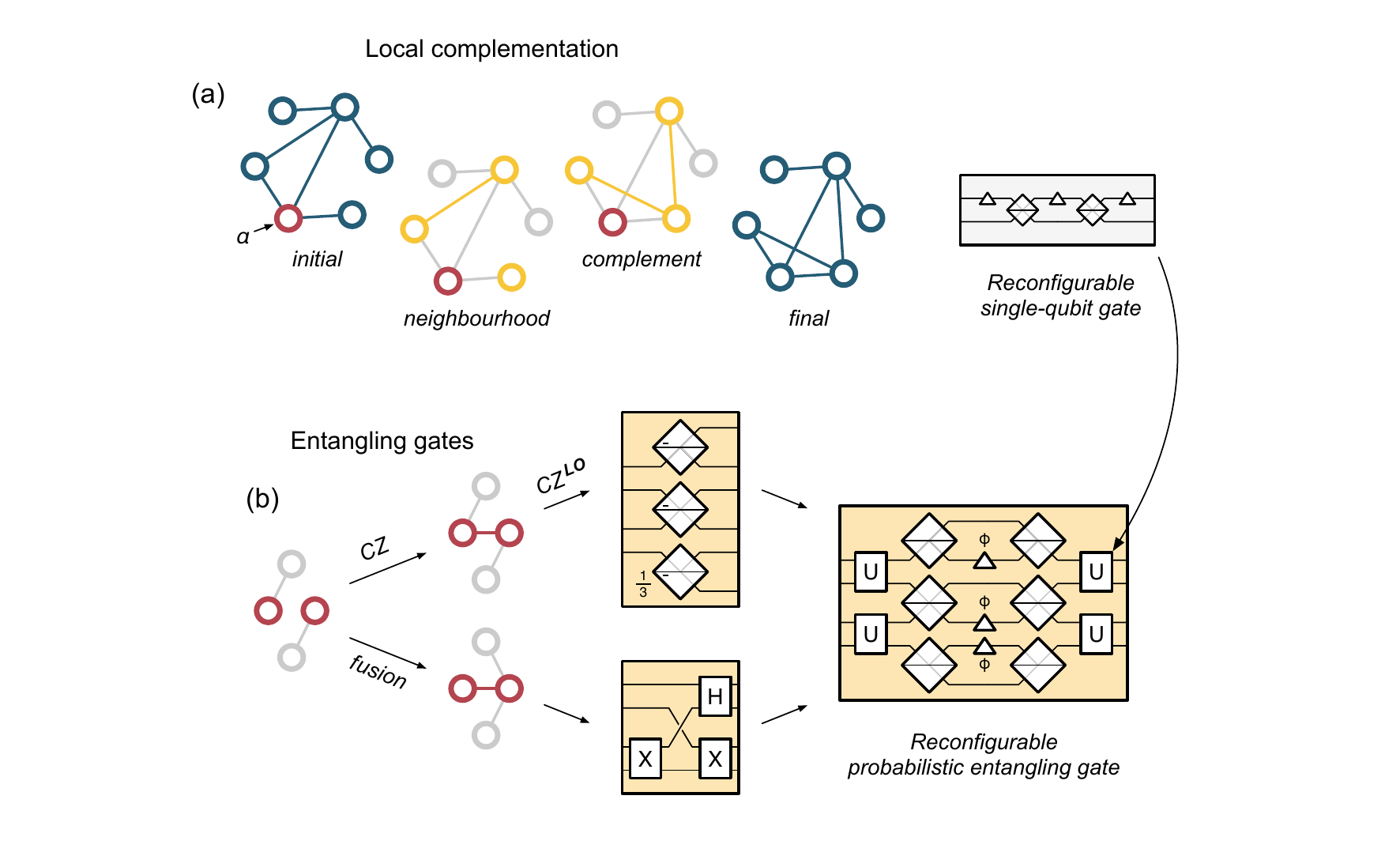}

\caption{\new{Operations on optical graph states. (a) Local complementation of qubit $\alpha$ on a six-qubit graph state. The subgraph containing only the neighbours of $\alpha$ is complemented: edges are removed if present, else added. Local operations, such as the local complementation, can be implemented by the single-qubit Mach-Zehnder interferometer (at right). (b) The two graph entangling operations and their implementations. \cz adds an edge to the graph, and fusion performs a vertex merge, and a vertex addition \cite{bodiya2006scalable}. The grey vertices may each be part of a larger graph. The reconfigurable postselected entangling gate, ($\text{R-PEG}$), introduced here, can perform both gates as well as local complementations, and so is universal for optical graph state postselection. Unlabelled beamsplitters are 50:50 and symmetric (with a $\pi/2$ phase on reflection); $1/3$ beamsplitters give a sign change on reflection from the side marked `$-$'.}}
\label{graphdef}
\end{figure*}

\subsection{Local complementation}

Local complementation provides a link between changes in the graph picture and local operations on the corresponding graph \emph{state}. Graphs which can be transformed into one another by successive applications of local complementation (LC) are locally equivalent---i.e. separated by single-qubit operations only \cite{hein2004multiparty, van2004graphical}. On a graph, $\text{LC}_{\alpha}$ acts to complement the neighbourhood of some vertex $\alpha$ (see figure \ref{graphdef}a). Specifically, successive application of the following local unitary, which implements $\text{LC}_{\alpha}$ on a graph $G$, can produce the entire set of states that are local unitary (LU) equivalent:
\begin{equation}
\text{LC}_{\alpha} = \sqrt{-iX_{\alpha}} \bigotimes_{i \in {N_G(\alpha)}} \sqrt{i Z_{i}}
\end{equation}
where $\sqrt{-i X} = \frac{1}{\sqrt{2}}\begin{psmallmatrix}1 & -i\\ -i & 1\end{psmallmatrix} $ and $\sqrt{i Z} = e^{\frac{i\pi}{4}}\begin{psmallmatrix}1 & 0\\ 0 & i\end{psmallmatrix} $. In optics, this local operation is implemented experimentally with a Mach-Zehnder interferometer between the two rails of qubit $\alpha$ (the neighbourhood of $\alpha$, $N_G(\alpha)$, must be known). We use $\text{LC}_{\alpha}(\ket{G})$, describing a unitary operation on  quantum state $\ket{G}$, and $\text{LC}_{\alpha}G$, the graph operation on graph $G$, interchangeably: $\ket{\text{LC}_{\alpha} (G)} = \text{LC}_{\alpha} \ket{G}$ . 

In graph-theoretic terms, local complementation realises, $\text{LC}_{\alpha} (G(V,E)):\rightarrow G(V,E')$, $E' =  E \cup K_{N_G(\alpha)} - E \cap K_{N_G(\alpha)} $, for some graph $G(V,E)$, where $K_{N_G(\alpha)}$ is the set of edges of the complete graph on the vertex set $N_G(\alpha)$. The subgraph $G[N_G(\alpha)]$, induced by the vertex set $N_G(\alpha)$, is complemented (has its edges toggled), leaving the rest of $G$ unchanged.

A repeated application of local complementation allows us to fully explore any class of locally equivalent graph states, given any member of that class \cite{hein2004multiparty, van2004graphical}. For a full treatment of single qubit operations on graph states, see refs. \tabcite{dahlberg2018transform} and \tabcite{dahlberg2018transforming}.

\begin{figure*}[t!]
\centering
\captionsetup{width=0.90\textwidth}
\includegraphics[width=0.55\textwidth]{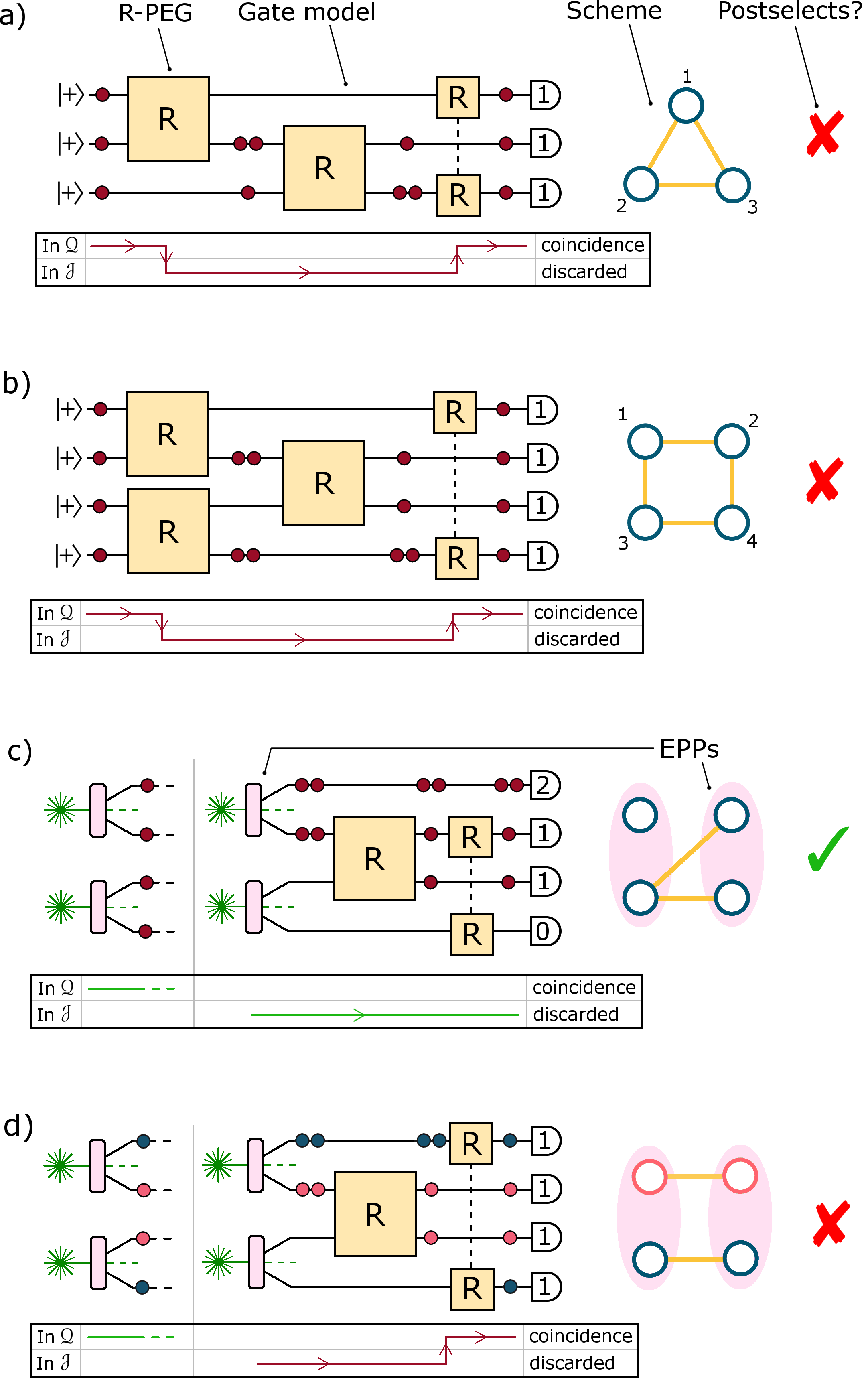}

\caption{Experiment schemes and postselection. a) A single Fock state vector traverses a scheme with a cyclic gate arrangement, getting placed back in the postselected qubit basis, $\mathcal{Q}$, by the final, cycle-inducing gate. We say the experiment does not postselect in this case, as gate failure is masked. b) Example of an experiment that involves concurrent initial gates not postselecting. Each initial gate produces photon-rich and photon-poor qubits which travel in opposite directions around the gate cycle, until the state re-enters $\mathcal{Q}$. c)  A source term in $\mathcal Q$ is shown next to a source term in $\mathcal J$. Junk states from degenerate EPP sources can effect postselectability, however this example does postselect, as there is no path back to $\mathcal{Q}$. d) Junk states from non-degenerate EPP sources re-entering $\mathcal{Q}$.}
\label{failures}
\end{figure*}

\subsection{Postselected entangling gates}

A gate configuration is postselectable if and only if all possible gate failure combinations can be detected, and ignored---when the postselected success signal is observed, all gates have performed as desired.

We distinguish two types of entangling gate in used linear optics---postselected gates (PEG) and heralded gates---both of which are probabilistic. Postselected gates have no auxiliary photons, and consume no photons on success. Heralded gates, on the other hand, either consume auxiliary photons as a resource (such as in the Knill-Laflamme-Milburn scheme \cite{knill2001scheme}), or consume one or more of the input-state photons (such as the cannonical fusion gate \cite{browne2005resource}). In both cases, the measurement outcome ``heralds'' the result of the gate. 

Using measurement and feed-forward, heralded gates remove any non-qubit components of the state, whereas PEGs produce a state which contains terms outside of the qubit subspace. These terms are ultimately filtered out by the measurement configuration, or in post-processing. Note that heralded gates with feedforward are sufficient for universal quantum computation \cite{knill2001scheme}, whilst postselected gates are not. Reference \tabcite{migdal2014multiphoton} discusses which photon-number state transformations are possible without postselection. Here, we only consider postselected gates.

PEGs are interferometers which couple modes between qubit mode pairs, implementing the desired operation on the qubit subspace, $\mathcal{Q}$. Components of the state in the ``junk'' non-qubit subspace, $\mathcal{J}$, are discarded. Here, $\mathcal{J} = \mathcal{F}  -  \mathcal{Q}$, where $\mathcal{F}$ is the space of all Fock states with $n$ or fewer photons, and $\mathcal{Q}$ is the qubit subspace, defined below. The output state of a PEG has components in both $\mathcal{Q}$ and in $\mathcal{J}$.

In the dual-rail encoding, a pair of optical Fock modes $f$ constitute each logical qubit $i$: $\ket{0}_i \leftrightarrow \ket{01}_f$, $\ket{1}_i \leftrightarrow \ket{10}_f$. Then $\mathcal{Q}_i = \text{span}(\{\ket{0}_i,\ket{1}_i\})$ and $\mathcal{Q}= \bigotimes_i \mathcal{Q}_i$. To postselect, we project on to $\mathcal{Q}$ with projector $P_\mathcal{Q}$.

We consider the postselected \cz (success probability $1/9$) \cite{hofmann2002quantum, ralph2002linear}, which we denote \czlo, and a postselected version of the fusion gate, $F$ (success probability $1/2$) \cite{bodiya2006scalable, browne2005resource}. These gates can both be implemented by the reconfigurable postselected entangling gate (\rpeg), shown in Figure \ref{graphdef}b. The \rpeg consists of three Mach-Zehnder interferometers (MZIs) over six modes.

Interferometers implement linear mode (Bogoliubov) transformations, which map Fock states to other Fock states unitarily. We can use postselection (via $P_\mathcal{Q}$) to understand the effect of \czlo and $F$ on qubit states. To understand their effect on qubit basis inputs (in $\mathcal{Q}$) we apply $P_\mathcal{Q}$ on the right-hand side as well. \czlo can be realised with the $\rpeg$ by setting $\phi =  \arccos(\frac{1}{3})$, yielding: 
		\begin{equation*}
			P_{\mathcal{Q}}  \czlo    P_{\mathcal{Q}} = \frac{1}{3} \cz,
		\end{equation*} 
equivalent to the cannonical \cz gate, after renormalisation. The postselected fusion gate, $F$, swaps $\ket{1}_1$ and $\ket{1}_2$ of its two input qubits, and applies a Hadamard gate to qubit 1, as shown in Figure \ref{graphdef}b \cite{bodiya2006scalable, bell2012experimental, yao2012observation}. This non-unitary operation deletes $\ket{01}_{12}$ or $\ket{10}_{12}$ qubit components of the state by transforming them to two-photon-per-qubit components in $\mathcal{J}$, which are filtered by postselection $P_\mathcal{Q}$ (coincidence detection). Fusion may be written in the qubit basis as:
		\begin{equation*}
			P_{\mathcal{Q}}  F  P_{\mathcal{Q}} = (\ddyad{+0}{00} \mkern+3mu + \mkern+3mu  \ddyad{-1}{11}).
		\end{equation*}
This is an entangling operation, and succeeds with probability $1/2$. For example, the action of this gate on the separable state $\ket{\text{++}}$ results in the (subnormalised) entangled two-qubit graph state $P_{\mathcal{Q}} F \ket{++} = \frac{1}{2}(\ket{+0}+\ket{-1})$. Figure \ref{graphdef}b describes the action of both \cz and $F$ gates on graphs.


\begin{figure*}[t!]
\centering
\captionsetup{width=0.90\textwidth}
\includegraphics[width=1.0\textwidth]{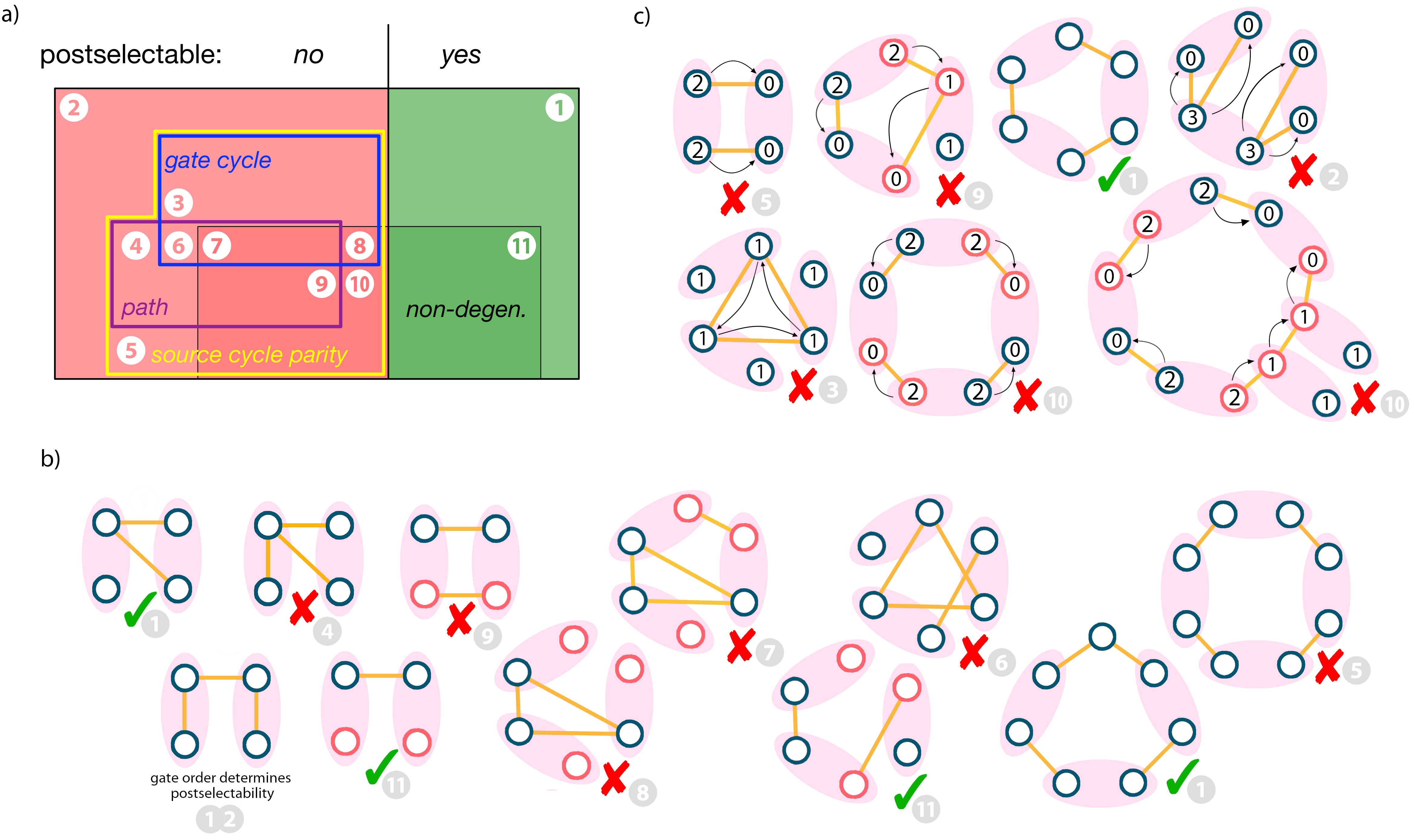}

\caption{\new{a) Venn diagram summarising postselection rules. The paths and gate cycles rules are necessary conditions, whilst the source cycle parity rule is sufficient only in the case of non-degenerate EPP sources. b) Elemental example schemes showing the postselectability of degenerate and non-degenerate EPP source experiments. The order that the gates are applied in does not effect postselectability, except where a gate can shuttle more than one photon at a time to cause failure---the example shown only postselects if the inter-source gate is performed last. c) Example schemes with source terms superimposed. Single photons travel through gates in the directions of the arrows to re-enter $\mathcal{Q}$}}
\label{venn}
\end{figure*}

\section{The limits of postselection}
\label{sec:pslimits}

In this section we demonstrate why certain arrangements of gates (PEGs) and photon-pair sources (EPPs) are not postselectable, and derive a simple test to inform postselectable experiment design. This is done by drawing the experiment as a graph, which we call a scheme.

\subsection{Criterion for gate postselectability}

Which gate arrangements are postselectable? Convention holds that sequential PEGs lead to gate failure, masked as success (i.e. are not postselectable). In this section, we derive a condition on the arrangement of these gates for this failure.

We wish to understand when a given combination of gates is not postselectable. To do so, we analyse the evolution of junk states (in $\mathcal{J}$) produced by PEGs. Since coincidence detection ensures that we only count output states with one photon in every qubit, we disregard parts of $\mathcal J$ with fewer than $n$ photons (for $n$ qubits in $2n$ modes). Here, we use ``qubit'' to refer to a pair of modes, whether they are occupied by photons, or not.

First, we discuss some properties of PEGs. When a gate fails, it produces an output state in $\mathcal{J}$, resulting in a qubit with no photons, and another with excess photons---one is photon-poor, and the other is photon-rich. Gates can also move states from $\mathcal J$ back into $\mathcal Q$. Generally, they can redistribute photons so that different qubits become photon-poor and photon-rich, throughout the circuit. If an arrangement of gates can move the state out of the qubit basis, $\mathcal Q$, and subsequently return the state to it, we say this gate arrangement is not postselectable, as the gate failure is masked (see Figure \ref{failures}a-b). Equally, the information of whether a the gate succeeded or failed is lost.

To re-enter $\mathcal Q$, the photon-poor and photon-rich qubits must meet again---after being generated together---by taking different paths through the experiment. Together, their paths form a loop. If we draw the gate configuration as a graph (with vertices as qubits, and edges as gates), then this loop corresponds to a cycle in that graph. We call this graph the experiment's scheme. Such a cycle is the only way the state can leave and subsequently re-enter $\mathcal{Q}$. Thus, configurations with cycles are not postselectable.

This holds for any time-ordering of the gates. In a cycle of gates, each gate's output is connected to an input of another gate. Junk is produced by all of the first time-step gates, re-entering $\mathcal Q$ by any subsequent gates that link them. A concrete example is shown in Figure \ref{failures}b.

As an example, if we apply three postselected \czlo gates to $\ket{\text{+++}}$, in an attempt to produce a three qubit ``triangle'' graph state (as in Figure \ref{failures}a) we find a corrupted result:
	\begin{equation*}
		\begin{split}
		\ket{\psi} &=   \mkern+5mu P_\mathcal{Q} \: \czlo_{31} \: \czlo_{23} \: \czlo_{12} \: \ket{\text{+++}}\\
		&=  \frac{1}{108} \left(\sqrt{2}+4 \sqrt{3}\right) \ket{000} \\
		&+     \frac{1}{54 \sqrt{2}}  \left(\ket{001} + \ket{010} - \ket{011}  + \ket{100} - \ket{101} - \ket{110} \right)  \\
		&+    \frac{1}{324} \left(4 \sqrt{3}-3 \sqrt{2}\right)\ket{111}  
		\end{split}
	\end{equation*}
The generated state, $\ket{\psi}$, differs significantly from the (unnormalised) desired state, which has equal weights of $1/(27 \cdot 2\sqrt{2})$ with signs $\{+,+,+,-,+,-,-,-\}$. Junk terms produced by earlier gates are returned to the postselected subspace by subsequent gates. The squared amplitude of this postselected state is well above the expected success probability: $ |\mkern-\thinmuskip\braket{\psi}\mkern-\thinmuskip|^2 \approx 0.00706 > (\frac{1}{9})^3  \approx 0.00137$; the state is dominated by terms which have re-entered $\mathcal{Q}$, scrambling the state. Cycles of fusion gates similarly scramble the state, resulting in states which are dominated by junk components which have re-entered $\mathcal{Q}$.

Therefore, a sufficient condition for successful gate postselection is: \emph{experiments containing cycles of PEGs are not postselectable}. We will refer to this as the ``gate cycles rule''. \new{In the next section, when postselected sources are considered, we will find that this gate cycles rule is encompassed by something altogether larger.}

\begin{figure*}[t!]
\centering
\captionsetup{width=0.90\textwidth}
\includegraphics[width=0.5\textwidth]{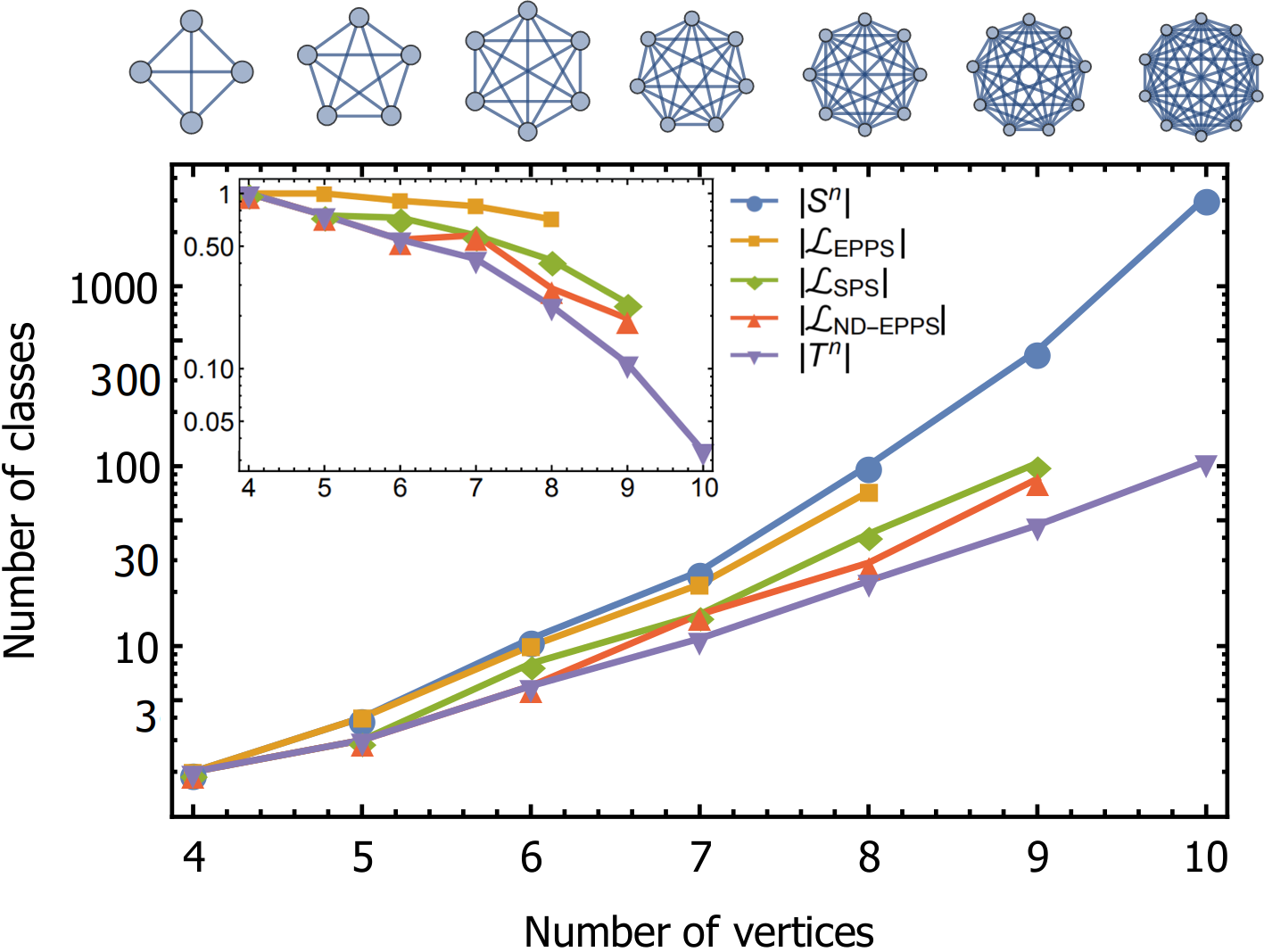}

\caption{Logarithmic plot of the number of graph state entanglement classes accessible to different resource states with $n$ qubits, $|\mathcal{L}_R|$
, as well as the total number of classes, $S^n$, and the number of trees with $n$ vertices $T^n$.  Here $\mathcal{L}_\text{EPP}$ is the set of classes accessible to degenerate entangled postselected pairs, $\mathcal{L}_\text{SPS}$ is the set of classes accessible to heralded single photon sources, and $\mathcal{L}_\text{ND-EPP}$ is the set of classes accessible to non-degenerate entangled postselected pairs. Inset: number of postselectable classes as a proportion of the number of entanglement classes. Isomorphic graphs are counted only once. Sequences from refs.\ \cite{danielsen2006classification, OEISA000055}.}
\label{plotgraphno}
\end{figure*}

We are unaware of any PEG that does not lead to two-photon-per-qubit terms, and indeed, such a gate would require there to be no photonic path between the modes of its input qubits. Such a gate will be postselectable even when used in cycles. To our knowledge, LCs, postselected CZ and postselected fusion represent the full known capability of postselective linear optics' to produce graphs states---all two-qubit Clifford gates can be decomposed into a CZ with LCs \cite{anders2006fast}. We also note that in ref. \tabcite{bodiya2006scalable} it is claimed that any graph state can be produced using postselected fusion only (with a focus on 2d lattice states). We can now see this claim to be unwarranted, since the proposed experiment violates the gate cycles rule.



\subsection{Combining degenerate postselected pair sources and gates}
\label{sec:degen}
In this section, we will examine experiments that utilise entangled postselected pair sources that produce degenerate (indistinguishable) pairs of photons. Ensembles of $m$ EPP sources produce states which are mostly in $\mathcal{J}$---for $m$ coherently pumped sources producing exactly $m$ pairs, a superposition of all $\begin{psmallmatrix}2m-1\\m\end{psmallmatrix}$ permutations are produced. 
Only one of these terms is in $\mathcal{Q}$---the term in which one pair is produced by each source---the rest are in $\mathcal{J}$ (see Appendix 1.5). Experiments combining EPPs and gates may not be postselectable even in the absence of cycles of gates, because the input superposition already contains junk. The gate cycles rule alone is insufficient in these situations. 
Experiments involving EPP sources can also be drawn as a graph, where both EPP sources and gates are represented by different types of edges. Example schemes are shown in In Figures \ref{failures}, \ref{venn} and \ref{alg}, where standard edges represent gates and pink ellipses represent source-edges.

The lowest-order junk state produced by $m$ pair sources is one with two photon-rich qubits (from source $i$) and two photon-poor qubits (from source $j\neq i$). These terms cause gates that disjointly connect the qubits from source $i$ to the qubits from source $j$ to mask gate failure. By disjointly, we mean that the two paths do not share gates (edges), however they may share vertices. The excess photons from source $i$ can travel via gates to the qubits of source $j$ and hence re-enter $\mathcal Q$ (causing postselection failure). This is depicted in Figure \ref{venn}b. Hence, \emph{experiments containing a pair of disjoint paths in the scheme that connect the qubits from one EPP source with the qubits from another, are not postselectable}. We will call this the ``paths rule''. Because each gate acts independently---simply moving a photon toward the photon-poor qubit---postselection fails no matter the gate order.


\begin{figure*}[t!]
\centering
\captionsetup{width=0.90\textwidth}
\includegraphics[width=1.0\textwidth]{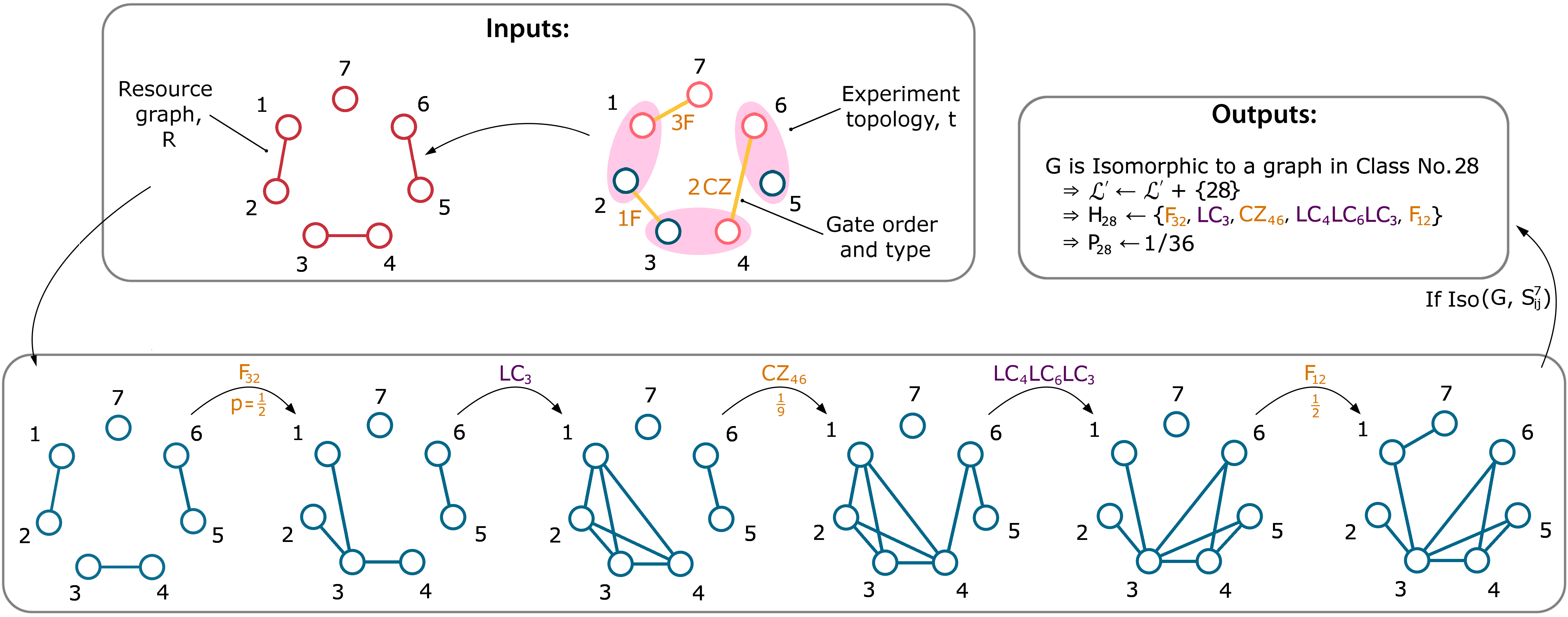}

\caption{An example of one Monte Carlo iteration of \textsc{FindAccessibleClasses}. Starting from a given resource state, in this case three non-degenerate EPPs and a heralded single photon, operations for an entangling strategy are performed in order, interspersed with \text{LC}s on the previously acted upon vertices. If the resulting graph is not isomorphic to any graph found thus far, the entanglement class of the graph $i$ is saved to a set $\mathcal{L}_R'$. After many runs, $\mathcal{L}_R'\approx \mathcal{L}_R$. An alternate example can be found in Appendix 1.2.}
\label{alg}
\end{figure*}

\new{When we consider the origins of the above paths rule, we observe that each source with one extra pair must simply connect with a source with zero pairs. This applies to any scheme which is a ring of $m$ EPP sources connected by gates---if $m$ is even, postselection fails. In schemes which are a ring of EPP sources, each qubit has just one gate acting on it. This means that higher-photon-number source terms ($n>2$) cannot re-enter $\mathcal{Q}$: $n-1$ connected gates are required to distribute $n$-photon terms to $n$ empty qubits (and bring the state back in to $\mathcal{Q}$). Hence only two-photon source terms can cause postselection to fail.}

\new{First we examine $m$ even. If alternate sources fire twice, then each source that didn't fire neighbours a photon-rich qubit. Hence the photons can travel through the gates and the can re-enter $\mathcal {Q}$: postselection fails. For a ring of $m$-odd connected sources, there can be no configuration where alternate sources fire, so postselection cannot fail in this way. For postselection to fail in a ring, any source that didn't fire must have both its qubits adjacent to others which fired twice. The term where alternate sources fire twice is the only one which satisfies this. Hence, postselection succeeds and we arrive at our third rule of postselection: \emph{Schemes that contain a cycle that has an even number of source-edges are not postselectable}. We will call this the ``source cycle parity rule''.}

\new{As with the paths rule, any number of gates can connect the sources, and these can be performed in any order---to fail, each gate simply moves one photon toward the photon-poor qubits. To be clear, gates may act on qubits from EPPs not involved in the cycle; both qubits of an EPP source must be part of the cycle to contribute to the rule (see Figure \ref{venn}).}

\new{The source cycle parity rule is a generalisation of the paths and gate cycles rules. The gate cycles rule is the case of zero sources (zero is even) and the paths rule is with two sources. Like the paths rule the source cycle parity rule is a necessary, but not sufficient condition of failure in the case of degenerate EPPs. The paths and source cycle parity rules only account for some sources firing twice, and others not firing at all. Some schemes, however, only fail to postselect because of higher-photon-number junk (see category 2 in Figure \ref{venn}). Furthermore, gates can not act independently if they can move more than one photon at a time, as is often the case with higher-photon number emission. In this case, gate order can effect postselectability. Hence the source cycle parity rule is a necessary, but not sufficient, condition of postselectability: an experiment which satisfies the it may still not postselect.}

Unfortunately, a sufficient condition for postselectability of combinations of PEGs and degenerate EPP sources is not forthcoming. In section 4, we apply numerical methods to evaluate postselectability, determining which graph states are postselectable for a variety of different source types.

\begin{figure*}[t!]
\centering
\captionsetup{width=0.90\textwidth}
\includegraphics[width=1.0\textwidth]{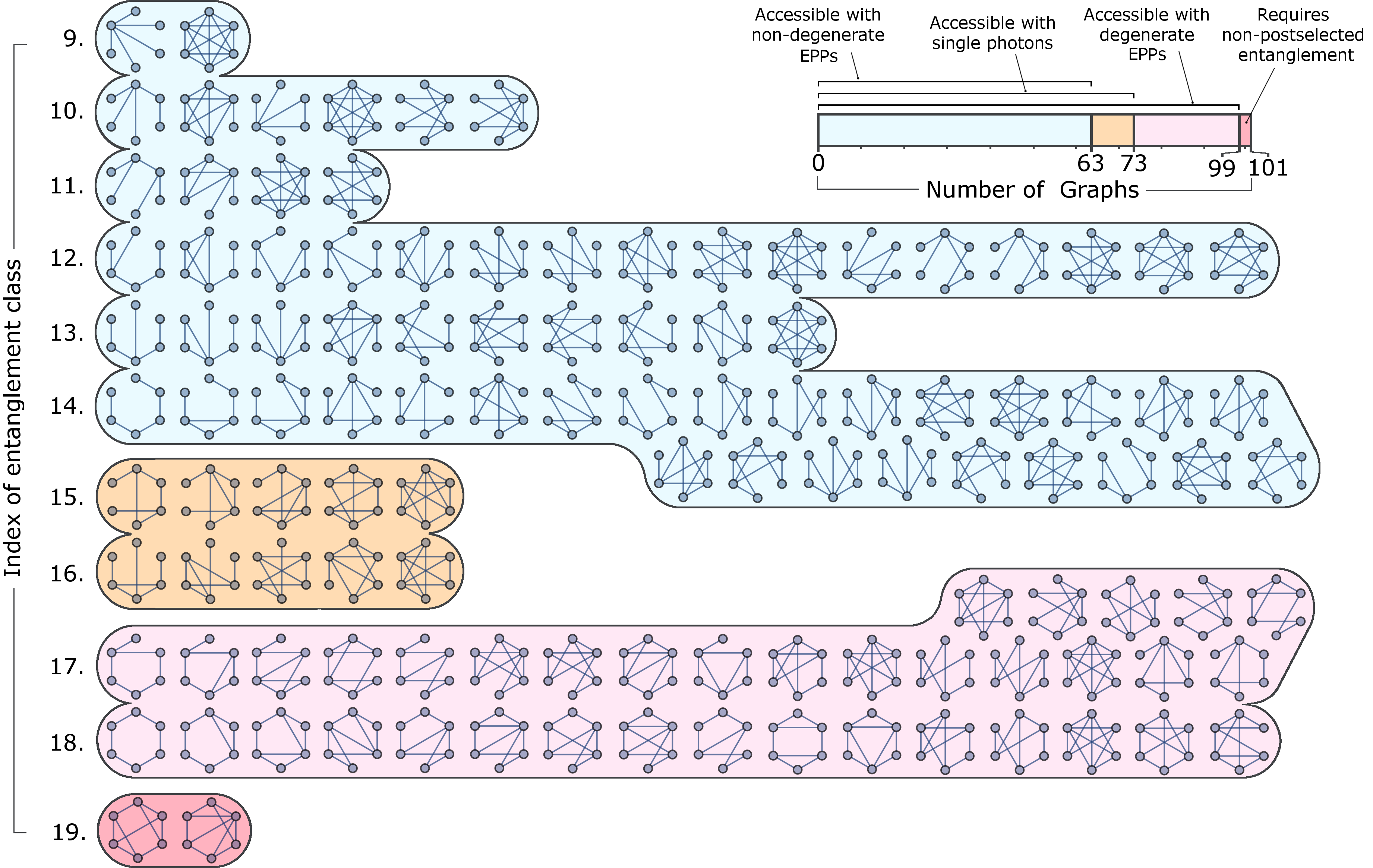}

\caption{Six photon graph states organised in to their LU classes, cannonically indexed starting from the 2-vertex graph state. See Appendix 1.3 for a complete list of graph states indexed up to 8 vertices. A resource of non-degenerate postselected pairs can access all of the graphs shown in blue, whilst a resource of single photons can access all of the states shown in orange and blue, and a resource of degenerate postselected pairs can access all of the graphs shown in pink, orange and blue. Graphs shown in red require a heralded entangled resource state to be produced using PEGs.}
\label{6qub}
\end{figure*}

\subsection{Non-degenerate nonlinear sources and postselected gates}

\new{An interesting and commonly used class of EPP source are those which generate non-degenerate photons, that is, photons of different colours, polarisations, etc. Since the two sets of photons are completely distinguishable, they cannot be interacted in gates. All interferometers that are not postselectable with degenerate EPP sources are also not postselectable with non-degenerate EPP sources, but number of valid gate arrangements is drastically reduced in the non-degenerate case. Here we will combine same-colour gates with the source cycle parity rule and show that in this case, it is a sufficient condition (see Figure \ref{venn}).}
	
\new{Firstly, schemes containing a cycle of $m$ odd  non-degenerate EPP sources (connected by gates) are not possible---the parity of the distinguishing degree of freedom prevents it. Combining this with the source cycle parity rule implies that any scheme containing a cycle of non-degenerate EPP sources will not postselect.}


\new{For sufficiency, we show that higher-photon-number ($n>2$) source terms cannot re-enter $\mathcal{Q}$ unless the source cycle parity rule is broken. With $m$ sources, the source cycle parity rule is broken for any scheme of more than $m-1$ gates---any connected graph of $2m$ vertices and $ e \geq 2m $ edges must contain a cycle, and the scheme has $m$ source-edges before any gates are added.}

\new{Take a connected scheme of $m$ non-degenerate EPP sources (and $2m$ qubits) where a higher-photon-number ($n \geq 2$) source term  re-enters $\mathcal{Q}$. Without loss of generality, this term is composed of sources $i=1,\ldots,K$ generating $s_i>2$ pairs, $n_1 = m-\sum_i^K s_i$ sources generating a single pair and $n_0 = \sum_i^K (s_i-1)$ sources generating nothing. To re-enter $\mathcal{Q}$, each of the $2n_0=2\sum_i^K (s_i-1)$ photon-poor qubits must be connected to a photon-rich qubit with at least one gate. Further, for the scheme to be globally connected, a gate must connect each of the $n_1=m-\sum_i^K s_i$ sources which fired once. Hence the scheme has $n_g \geq 2n_{0} +n_{1} =2\sum_i^K (s_i-1) + m - \sum_i^K s_i = \sum_i^K s_i - 2K +m$ gates. Since $s_i \geq 2$ we have $n_g > m-1$ and the scheme breaks the source cycle parity rule.}





\new{In Section \ref{sec:degen} we find that gate order has no effect on postselectability for the source cycle parity rule. Hence our graphical schemes, which do not encode a gate order, are are enough to decide postselectability for non-degenereate EPPs experiments.}

\begin{table*}[t!]
\centering
\captionsetup{width=0.90\textwidth}
\small{

\begin{adjustbox}{center}
\begin{tabular}{ p{0.1cm} p{6cm} P{8.0cm} P{2.2cm}}

 $n$ & Resource state, $R$                   &  Indices of LU classes accessed, $\mathcal{L}_R '$         &  $|\mathcal{L}_R '| / |S^n|$   \\
 
\noalign{\vskip 1mm}    
\hline
\hline
\noalign{\vskip 2mm}    
 $4$ & 2 non-degenerate EPPs                                  				&    3, 4                                          & $2 / 2 = 1$ \vspace{0.2cm}\\

 $5$ & 2 non-degenerate EPPs \& 1 single photon           					&    5 to 8                                          & $3 / 4 = 0.75$ \\
 $5$ & 5 single photons                              					&    5 to 8                                          & $3 / 4 = 0.75$ \\
 $5$ & 2 degenerate EPPs \& 1 single photon              				    &    5 to 9                                          & $4 / 4 = 1$ \vspace{0.2cm}\\

 $6$ & 3 non-degenerate EPPs				                            &    9 to 14                                         & $6 / 11 \approx 0.54$ \\
 $6$ & 6 single photons                                   				&    9 to 16                                         & $8 / 11 \approx 0.73$ \\
 $6$ & 3 degenerate EPPs								 				&    9 to 18                                         & $10 / 11 \approx  0.91$ \\
 $6$ & 2 pairs \& 2 single photons                        			&    9 to 19                                         & $11  / 11 = 1$ \vspace{0.2cm}\\

 $7$ & 3 non-degenerate EPPs \& 1 single photon				      	  &     20 to 32, 34, 36                               & $15 / 26 \approx 0.58$ \\
 $7$ & 7 single photons                                                   &     20 to 32, 34, 36                               & $15 / 26 \approx 0.58$ \\
 $7$ & 1 pair \& 5 single photons                                     &     20 to 41                                       & $22 / 26 \approx 0.85$ \\
 $7$ & 3 degenerate EPPs \& one single photon					  		  &     20 to 41                                       & $22 / 26 \approx 0.85$ \\
 $7$ & 2 entangled pairs \& 3 single photons                                   &     20 to 45                                       & $26 / 26 = 1$ \vspace{0.2cm}\\

 $8$ & 4 non-degenerate EPPs        						&   46 to 68, 70, 73, 74, 79, 89, 92                                   & $29 / 101 \approx0.29$ \\
 $8$ & 8 single photons                 					&   46 to 79, 83, 86 to 87, 89, 92, 94, 104, 107, 121, 143             & $42 / 101 \approx 0.42$ \\
 $8$ & 1 pair \& 6 Single Photons      						&   46 to 108, 114, 115, 117, 121, 123 to 125                         & $73 / 101 \approx0.72$ \\
$8$ & 4 degenerate EPPs		&   46 to 68, 70, 75, 76, 78 to 80, 82, 83, 86 to 89, 92 to 98, 100, 101, 103, 108, 109, 110, 112 to 116, 118 to 121, 123 to 125, 127, 129, 130, 133, 134, 136, 139, 141, 144, 145                                         			 	       & $72 / 101 = 0.71$ \\
 $8$ & 4 entangled pairs                        			&   46 to 140, 142 to 144, 146                                         & $99 / 101 \approx0.98$ \vspace{0.4cm}\\

$9$   & 4 non-degenerate EPP \& one single photon       &   148 to 197, 199, 200, 204 to 208, 212, 214, 218, 219, 221, 222, 225, 227, 233 to 235, 238, 241, 247, 249 to 251, 256, 259 to 260, 263, 272, 274, 275, 297, 328, 392, 410                              &   $85 / 440 \approx 0.19$    \\ 
$9$   & 9 single photons                                 &   148 to 200, 202, 204 to 208, 211, 212, 214 to 215, 218, 219, 221 to 222, 225, 227 to 229, 233 to 235, 237 to 238, 241, 247, 249 to 251, 256, 259, 260, 262, 263, 271 to 275, 290, 297, 299, 301, 305, 319, 328, 329, 342, 346, 388, 392, 401, 41         & $ 104 / 440\approx 0.24$  \\ 
$9$  &4 pairs \& one single photon                   &   148 to 558, 560 to 567, 569 to 572, 575 to 578, 581, 584, 586, 587                      &   $ 431 / 440\approx 0.98$ \vspace{0.0cm}\\

\end{tabular}

\end{adjustbox}
}

\caption{Entanglement class accessibility for different resources. $\mathcal{L}_R '$ is the set of classes of graph state that can be generated using PEGs given different resource states $R$. Entanglement classes (denoted $S^n_i$) are indexed starting from the connected $2$-vertex graph state. See Appendix 1.3 for a complete list of graph states indexed up to 8 vertices (from refs.\ \cite{cabello2011optimal, hein2004multiparty, danielsen2006classification}). This table was generated by \sc{FindAccessibleClasses}.}
\label{table1}

\end{table*}


\section{Which graph states are postselectable?}
\label{sec:whichps}

Now that we are familiar with the rules of postselecting graph states, we can establish which states can be accessed, and which states cannot.

First, we lower bound the number of classes that are accessible from the popular resource state of $n$ non-degenerate pair sources, with the number of trees (non-cyclic graphs). This follows from the fact that trees can always be constructed from this resource (see Appendix 1.1), and that there is at most one tree in each entanglement class \cite{bouchet1988transforming}. Figure \ref{plotgraphno} compares the number of trees to the number of entanglement classes for increasing qubit (vertex) number, and reveals a super-exponential divergence \cite{OEISA000055, bergeron1998combinatorial, danielsen2006classification}. 

Because of the combinatoral number of possible experiments using PEGs and LCs, we turn to numerical methods to discover exactly which classes of entanglement are accessible to postselective linear optics given a certain resource state. Our approach is to sample allowed combinations of entangling gates and local complementations, and catalogue the graph states which result. By using LCs, postselected CZ and fusion gates, we span the currently known capability of postselective linear optics' to produce graphs states; though gates can only be applied in trees, the use of LCs allows access to a wider variety of graph state classes---including those not containing trees. Note that all two-qubit Clifford gates can be decomposed into a CZ with LCs \cite{anders2006fast}. We use the canonical indexing provided by Hein, Cabello, \emph{et al.} in refs. \tabcite{hein2004multiparty} and \tabcite{cabello2011optimal}. We denote the set of graph state class indices that can be accessed by a given resource state $R$ by $\mathcal{L}_R$. In this section we establish methods to find which graph states are postselectable using different kinds of photon source. The results of this are shown in Table \ref{table1}, whilst a complete list six-qubit graph states, showing which are postselectable is shown in Figure \ref{6qub}.

\subsection{Numerical Methods}

In \cite{danielsen2006classification, cabello2011optimal}, tables of representative members for each entanglement class up to $n=12$ are provided. Starting from these supplied graphs, we take random walks to explore each LU class (which are of known size). We denote the $j^\text{th}$ $n$-vertex graph of entanglement class $i$ as $S^n_{ij}$, where $S^n_{i}$ is the set of all graphs in that entanglement class and $S^n$ is the set of all $n$-vertex classes, $S^n = \cup_i S^n_i$. Note $S^n_a \cap S^n_b \neq \emptyset \quad \forall a,b$.

We explore which graph states are accessible to linear optics and PEGs with our algorithm \textsc{FindAccessibleClasses} (visually depicted in Figure \ref{alg} and provided in Appendix 1.2). This algorithm enumerates accessible entanglement classes for a given resource state $R$, and stores the list of accessible classes as $\mathcal{L}_R' \approx \mathcal{L}_R$. The results of our investigation using \textsc{FindAccessibleClasses} are shown in Table \ref{table1}. To aid the classification of any newly found graphs, $S^n$ is stored in memory. We provide plain-text and Mathematica representations of $S^n$, up to $n=9$ qubits, as well as a Mathematica implementation of \textsc{FindAccessibleClasses}. These can be found in the Appendix \cite{adcock2018supp}.

Each iteration of \textsc{FindAccessibleClasses} starts with a resource graph state $R$, and a randomly chosen $n$-qubit schemesatisfying out postselection rules. This $n$-vertex gate arrangement is a graph, $t$, which corresponds to a linear optical experiment composed of PEGs. $t$ is generated randomly from a set $T$, which depends on the resource (see Section 4.2). Each edge of $t$ is randomly assigned either a \cz or an $F$ gate, and a random time ordering. Random combinations of relevant \text{LC}s are interspersed between gates (see section 4.3). As the state transformations are encoded as operations on graphs, we avoid the exponential memory requirement of simulating quantum states. This Monte-Carlo approach is designed to fairly sample postselected linear optical graph-state generation experiments.

After all specified entangling operations have been performed, we store the entanglement class that has been reached, in $\mathcal{L}_R'$. We store the gates and local complementations used, in $H^R$ ($H^R_i$ is the \emph{recipe} for accessing entanglement class $i$ with resource state $R$). These recipes are overwritten when a more efficient one is found. We keep sampling, until no more novel classes are found, at which point we assume we have sampled them all: $\mathcal{L}_R'\approx \mathcal{L}_R$.

\subsection{Resource States}

To fairly sample from postselectable experiments, the sets of allowed gate arrangements, $T$, varies for the different classes of resource: heralded states, degenerate EPPs and non-degenerate EPPs. In all cases, relevant LCs are interspersed between the gates.

For a heralded resource (one not involving postselected sources), our algorithm must sample from all possible experiments that do not contain a cycle of PEGs. Hence $T$ is the set of all connected trees with $n$ vertices. There are $2^{n-1} n^{n-2} (n-1)!$ such experiments. This is the number of labelled trees \cite{cayley1889theorem} multiplied by both the number of possible edge labellings $\{\cz, \text{F}\}$ and the number of gate orderings. Labelled trees are fairly sampled by the Pruefer sequence method \cite{prufer1918neuer}. We sample from the set of all isomorphisms (labellings) of every tree, as to consider only one labelling of the resource state.

For a heralded resource that has some entanglement, not all of the gates represented by $t \in T$ need be applied to obtain $n$-partite entanglement. In \textsc{FindAccessibleClasses}, the number of gates applied are chosen at random, where the minimum number of gates could still feasibly output an $n$-qubit connected graph. For example starting with 3 heralded entangled pairs, only two gates are needed to produce a 6-qubit state, however more entanglement classes are accessible by using up to $5$ gates. Hence, we randomise over the different number of gates. An example member of $t \in T$ is shown in Figure \ref{alg}, with its ordered edge-labelling. The set of trial gate arrangement trees, $T$, can also be tailored to encode any restriction in gate topology, for example where only nearest-neighbour gates are permitted.

For schemes with $m$ non-degenerate EPP sources, $m-1$ gates are required to globally entangle the state. Further, these are not postselectable with $m$ gates, as a cycle will form between sources. In this case, the number of gates is fixed, and we only sample over the trees of order $m$. Since $m=n/2$, and the number of trees scales exponentially with $n$, the search space for non-degenerate pairs is drastically smaller than for any other input resource.

In the case of a resource state of degenerate EPPs, postselectability of a particular gate arrangement must be evaluated directly, since we have no sufficient postselection rule rely on in this case. We evaluate many random (non-interacting) photon scattering through the gate arrangement, for each source term in $\mathcal{J}$. If the photons can return to a one-per-qubit state, in $\mathcal{Q}$, then the gate combination is discarded, as it is unpostselectable. Although this method is not exhaustive, a sufficiently large number of iterations can guarantee accuracy. Due to the extra cost associated with this subroutine, \textsc{FindAccessibleClasses} evaluates the postselectability of one experiment, then permutes the choice of gate (CZ or F), as well as the LCs applied, for 50 experiments which share the topology. This greatly increases the efficiency of class discovery, and is analogous to finding all of the classes accessible by a single scheme which utilises the reconfigurable PEG of Figure \ref{graphdef}b, then moving to another scheme.

\subsection{Sampling linear optical graph experiments}

Applying local operations between gates can yield a wider variety of accessible graph states. Just a few \text{LC} operations are sufficient. After a $\cz_{ij}$ operation, \text{LC}s need only be applied to qubits $i$ and $j$, since \cz commutes with \text{LC} in all other cases. After a $F_{ij}$ operation, \text{LC}s need only be applied to qubits in $\{N_G(i) \cup  N_G(j) + \{i\} + \{j\}\}$, for the same reason, where $N_G(i)$ is the graph neighbourhood of vertex and qubit $i$. Furthermore, $\leq 5$ randomly chosen \text{LC} operations are needed after a \cz gate, since \text{LC} around two vertices is periodic with period 6 \cite{danielsen2008edge}. Similarly, $ \leq 14$ \text{LC}s are needed after a fusion, as this is the largest number of \text{LC} operations needed to traverse the widest class for $n\leq9$ qubits (from numerics). Higher qubit numbers will require concomitantly larger numbers of post-fusion LCs. Proofs are shown in Appendix 1.1.

The size of the configuration space for a resource of heralded single photons is $O(2^{n-1}d_n^{n-2}n^{n-2}(n-1)!)$, where $d_n$ is the diameter of the orbit of the largest $n$-qubit entanglement classes ($d_9=14$). For $n=8$ qubits the size of this configuration space is $\approx 1.3\times 10^{18}$; for $n=10$ it explodes to $\approx 2.7\times 10^{25}$. This makes an exhaustive search impossible, and motivates our use of sampling methods.

Each newly found $n$-vertex graph, $G$, is likely to be isomorphic, not identical, to the corresponding graph in the stored database of entanglement classes, $S^n$. Consequently, \textsc{GraphIsomorphism}, which is computationally hard, must be computed for a range of graphs so that the candidate graph can be properly catalogued. For small numbers of vertices, however, the problem is tractable.

\subsection{Which classes are accessible?}

We have enumerated the entanglement classes that are accessible using certain resource states and PEGs. The results are shown in Table \ref{table1}. Interestingly, those states which are inaccessible tend to have higher canonical indices, ordered both by vertex degree, and by the minimum number of edges on a graph in the class. This also correlates with known bounds on the Schmidt rank \cite{hein2004multiparty, cabello2011optimal}.

Without exploring the entire space, there is no guarantee that all classes have been found, however \textsc{FindAccessibleClasses} appears to converge after sampling a minuscule fraction of configurations. For the results in Table \ref{table1}, the algorithm was terminated when no new classes were found in the last 5/6 of the total number of iterations. For the 8-qubit graphs, this corresponds to sampling about one in every $10^{5}$ possible configurations.

Our exploration using $n\leq9$ was performed using the (un-compiled, single-threaded) Mathematica implementation, on a standard desktop PC. We anticipate that, by using parallel, compiled code, up to $n=12$ vertex graphs can be investigated.

\subsection{Designing graph state generators}

Efficient experimental generation of graph states using photonics is a challenge. \textsc{FindAccessibleClasses} can also be used for experiment design. To find a recipe which produces the desired graph state in class $i$, repeat \textsc{FindAccessibleClasses} with a reasonable resource state input $R$, maximising the probability of generation, $P^R_i$, and recording the recipe, $H^R_i$. Modify $R$ until a satisfactory experiment is found: increase $P^R_i$ by using a more practical $R$ or increase it (especially in the case where no graphs in $i$ are found) by using a more entangled $R$. The recipe $H^R_i$ yields an experiment which produces a graph in the target class $i$. using an $\rpeg$ to realise each PEG (see Figure \ref{graphdef}), where the single-qubit interferometers can perform the necessary local complementations, the circuit is simple to assemble. The optical depth of such circuits is $O(n)$ for $n$ qubits. 

Conversely, \textsc{FindAccessibleClasses} can enumerate the graph states that a given interferometer can access, by fixing $t$ and searching over different combinations of LC and gate type. In this way, combining the rules of postselection, and \textsc{FindAccessibleClasses} and the R-PEG allows new graph generating experiments to rapidly checked for feasibility, and designed with maximum versatility.

\section{Conclusion}
\label{sec:conclusion}

\new{Postselection remains the current go-to tool for generating entanglement between photons in linear optics. We have described a severe and previously undocumented limitation of this technique, in the context of both heralded sources, and degenerate and non-degenerate EPP sources. By drawing gate and source arrangements (schemes) as graphs, where the qubits are vertices, gates are one type of edge and sources are another, we derive three simple rules for postselection:}

\setlist[description]{font=\normalfont\itshape\textbullet\space}
\begin{description}
\item [Gate cycles rule:] \textcolor{Black}{Schemes that have a cycle of gates will not postselect. For deterministic sources of photons, this condition is sufficient.}

\item [Paths rule:] \textcolor{Black}{For experiments with EPP sources. Schemes which disjointly connect both qubits of two sources with one another will not postselect.}
	
\item [Source cycle parity rule:] \textcolor{Black}{For experiments with EPP sources. Schemes which contain a even number of source-edges in a cycle, will not postselect. This subsumes the above rules as cases of zero and two sources. This is a necessary condition in the case of degenerate EPPs, and a sufficient one in the case of non-degenerate EPPs.}
\end{description}

Further, we have tabulated which graph states can be accessed using linear optics and PEGs for $n\leq 9 $, and demonstrated that the number of states available to postselected systems diminishes rapidly with increasing qubit number. We have provided algorithms to calculate which states are accessible for graphs up to $n=12$, limited by the availability of class catalogue data. 

Postselectable graph states that can be produced from EPPs are of interest, with experiments with up to 12 photons possible in the near future \cite{wang2016experimental}. Whilst we have not found an analytic condition on postselectability for the combined postselection of degenerate EPPs and PEGs, our numerics show that a wide variety of states are nonetheless available to this resource---the majority of classes are accessible for $n\leq 8$, though this fraction diminishes with increasing $n$. This is encouraging news for near-term demonstrations of entanglement using linear optics. Postselected sources and gates may still have some mileage before true single photon sources become a necessity.

Despite their widespread use, a vanishing fraction of graph states are accessible using non-degenerate postselected pairs, and these accessible states tend to have low Schmidt rank. With heralded single photons, many more LU classes become accessible, but this too is a diminishing fraction of the total, as qubit number increases. The end of the road for postselected quantum optics is now in sight. Heralded or deterministic gates for photon-photon interactions are not just a route to increased efficiency, but are a necessity if we are to access any appreciable fraction of multi-qubit entanglement classes using optics.

Several questions remain unanswered. Why are certain states accessible and others are not? Why does interspersed local complementation allow for the creation of a wider variety of states? What is the size of the space accessible to hybrid experiments, part postselected, part heralded? Can this reasoning be applied to hyper-entangled, or qudit photonic states? Is there a sufficient rule for the postselection of schemes of degenerate EPP sources?

The limits of postselection are indeed severe, but, with the tools and understanding developed here, planning quantum information experiments which reach these limits will be possible. Multi-photon experiments are often phrased with a measurement-based or state-preparation focus, both of which are enlightened by this work, in the context of postselection. These methods will allow experimenters to produce states with the minimum resource, and with the most efficient optical recipe, expediting progress toward large-scale quantum computation, with optics and otherwise.

\section*{Acknowledgements} 
The authors would like to acknowledge Mercedes Gimeno-Segovia, Sam Pallister, Patrick M. Birchall, Stasja Stanisic, Will McCutcheon, and Raffaele Santagati for fruitful discussion and motivation. This work was supported by the UK Engineering and Physical Sciences Research Council (EPSRC). JCA and SMS are supported by EPSRC grant EP/L015730/1. JWS is supported by EPSRC grant and EP/L024020/1.

\printbibliography

\clearpage
\setcounter{section}{1}
\section*{Appendix}

\subsection{Proofs}

 \begin{lemma}
All of the $n$-vertex graph states that are locally equivalent to a tree can be constructed from $\ceil{\frac{n}{2}}$ entangled postselected pairs (from postselected nonlinear pair sources) using only postselected \text{CZ} and fusion gates. 
 \end{lemma}
 \begin{proof}
 
The two distinct trees with $n=4$ vertices (the ``star'' and ``line'' graphs) correspond to performing $\text{CZ}^{LO}$ or fusion on two pairs respectively, which establishes the base case. The induction step is to show for any order $n+2$ tree, $t$, one can always find a feature of the tree that implies it could have been constructed from some order $n$ tree, and the two vertex connected graph (entangled pair), using an egde-add (CZ) or fusion operation.

These features of $t$ are as follows:

\begin{description}[align=right,labelwidth=2cm]
\item [Feature 1:] $t$ has two vertices in a line formation, where the second vertex is only adjacent to the first. This corresponds to a CZ (edge-add) of an order $n$ tree with the complete two-qubit graph.

\item [Feature 2:] Two leafs (vertices of degree 1) are adjacent to the same vertex of $t$ (but no others). This corresponds to a fusion of an order $n$ tree with the complete two-qubit graph.
\end{description}

To show that all trees have one of these Features, we perform another induction. All order $n+1$ trees can be constructed by adding a vertex (with connecting edge) to some $n$-vertex tree. We will show that these Features can disappear when adding a connected vertex, but only by creating the other Feature. Hence all trees have at least one of these features.

Feature 1 will disappear if a new vertex is connected to the degree-two vertex of the Feature. In this case, the new graph has Feature 2. Similarly, Feature 2 will disappear if a new vertex is connected to one of the vertices of Feature 1. This forms Feature 1. The only tree of three vertices has both of these features. Hence all trees have one of these two features. 

Since all $n+2$ trees have one of these features, it is always possible to find a tree of order $n$ that can be used to construct a tree of order $n+2$, down to $n=4$ where we know how to make both of the trees.

Since each additional pair of vertices has only one gate acting on it, the postselection rules are not violated.

This completes the proof.

 \end{proof}

\begin{lemma}
$[\text{CZ}_{ij}, \text{LC}_\alpha ]=0$ $\forall$ $\alpha \notin \{i,j\}$. ($\text{CZ}_{ij}$ commutes with \text{LC} applied to qubit $\alpha$, $\text{LC}_\alpha$ when $\alpha$ is not one of the qubits acted upon by the \text{CZ}.)
\end{lemma}

\begin{proof}

If $i,j \notin N_G(\alpha)$ the unitaries (graph operations) act on different qubits (vertices) and therefore commute.  We now examine $i,j \in N_G(\alpha)$. Note that complementation of a subgraph defined by a fixed set of vertices commutes with a \text{CZ} operation, since both toggle the edges present in the graph (addition modulo 2). This can also be understood by examining the \text{CZ} and \text{LC} unitaries. In $\text{LC}_\alpha$ for  $i,j \in N(\alpha)$, qubits $i$ and $j$ undergo a $\sqrt{iZ}$ operation, which is diagonal. Since \text{CZ} is also diagonal, these operations commute, that is  $[\text{CZ}_{ij}, \sqrt{iZ_k} ]=0$ for $k=i,j$. Note that $\sqrt{iZ_k}\otimes \sqrt{iZ_l}$ is also diagonal. Since $N_G(\alpha)$, is unaffected by the \text{CZ}, $\text{LC}_\alpha$ and \text{CZ} commute if $i,j \notin N(\alpha)$. 

\end{proof}

\begin{lemma}
Repeated \text{LC} on $i$ and $j$ on some graph $G$ has just one periodic path through the members of the LC class. 
\end{lemma}

\begin{proof}

This is demonstrated independently in terms of edge-local complementation in \cite{danielsen2008edge}, but we provide an alternate proof. Since $\text{LC}_{\alpha}  \circ  \text{LC}_{\alpha} = \mathds{1} $ there are only two ways uniquely apply LCs---alternating \text{LC} on $i$ and $j$, i.e. $\dots \text{LC}_i  \circ  \circ \text{LC}_j$ and $\dots  \text{LC}_i  \circ \text{LC}_LC_j$. This defines two paths through the members of the LC class.

We will now show these paths are periodic. In the following we denote the $k^{th}$ LC operation of one of these trajectories as, $\text{LC}^{k}$, 

Since there are a finite number of graphs equivalent under LC, alternating LCs must reproduce the initial graph, or the path will end after $k-1$ LCs, i.e. when some graph is reached whereby LC has no effect. In this case, $\text{LC}^{k}= \mathds{1}$ and the series of LCs can be written $\ldots \text{LC}^{k+1}_i \circ \text{LC}^{k}_j  \circ \text{LC}^{k-1}_i \circ \ldots \circ \text{LC}^{1}_j = \ldots \text{LC}^{k+1}_i \circ \text{LC}^{k-1}_i \circ \ldots \circ \text{LC}^{1}_j$. Since $\text{LC}_i \circ \text{LC}_i = \mathds{1}$, all LCs can be paired around $\text{LC}^k$ and cancelled, leaving the identity, i.e. the operation is periodic with period $2k-1$.

We will now show that these orbits are the inverse of one another.  Take one of the two orbits (say the one that starts with $j$) and assume it has period $p$, then $\text{LC}^{p}_i \circ \text{LC}^{p-1}_j \circ \ldots \circ \text{LC}^1_j \> (G) = G$. Applying $LC_i$ here we find $\text{LC}^{1'}_i \> G =  \text{LC} ^{1'}_i \circ  \text{LC}^{p}_i  \circ \text{LC}^{p-1}_j  \circ\ldots \circ \text{LC}^1_j \> (G)  = \text{LC}^{p-1}_j \circ \ldots \circ \text{LC}^1_j \> (G) $. Similarly to the above, each operation in the second orbit inverts an operation in the first orbit until we arrive back at the starting graph $G$. Hence, the second path, (begininning with $\text{LC}_i$ is the reverse or the first.

\end{proof}

This implies that only one trajectory need be considered in \textsc{FindAccessibleClasses}. We henceforth always start the orbit with $\text{LC}_j$, which we will denote $\text{LC}^k$ for the $k^{th}$ \text{LC} of an orbit.

We now prove that such an orbit has period at most 6 for all graphs, independent of of the number of vertices.

\begin{lemma}
Repeated application of \text{LC} on vertices $i$ and $j$ on some graph $G$ has period at most 6, independent of the number of vertices. 
\end{lemma}

\begin{proof}

This is demonstrated independently in terms of edge-local complementation in \cite{danielsen2008edge}, but we provide an alternate proof. Starting with some graph state $G=G^0$ let $G^k$ be the graph state after $k$ \text{LC} s. Next, we define three sets of vertices. 

Firstly, The set of vertices which are in the neighbourhood of $i$, but not in the neighbourhood of $j$ and excluding $j$ and $i$, which we label 
$$
\mathcal{X}^k = \{N_{G^k}(i) - N_{G^k}(j) - \{j\} -\{i\}\}
$$
Secondly, the intersection of vertices which are in the neighbourhood of $i$ and the neighbourhood of $j$. 
$$
\mathcal{Y}^k = \{N_{G^k}(i)\cap N_{G^k}(j)\}
$$
And finally the set of vertices which are in the neighbourhood of $j$, but not in the neighbourhood of $i$ and excluding $i$ and $j$. 
$$
\mathcal{Z}^k = \{N_{G^k}(j) - N_{G^k}(i) - \{i\} - \{j\}\}
$$

Where for clarity we write $\mathcal{X}^0=\mathcal{X}$, $\mathcal{Y}^0=\mathcal{Y}$,  $\mathcal{Z}^0=\mathcal{Z}$. Now we can examine the effect of the \text{LC} orbit on a graph $G=G^0$.

Since $i$ and $j$ are always neighbours, the effect of $\text{LC}_i$ on the $k^{th}$ member of the orbit is to swap the sets $\mathcal{X}^{k}$ and $\mathcal{Y}^{k}$, and the effect of the $\text{LC}_j$  is to swap sets $\mathcal{Y}^{k}$ and $\mathcal{Z}^{k}$, yielding the following, for $k=1,\dots,p$:

\begin{itemize}
\item For odd $k$, $\text{LC}^k$ complements the subgraph induced by the sets $N_{G^k}(j)=\mathcal{Y}^k\cap\mathcal{Z}^k + \{i\}$, setting $\mathcal{Y}^{k+1} = \mathcal{Z}^{k}$ and $\mathcal{Z}^{k+1} = \mathcal{Y}^{k}$. 

\item For even $k$, $\text{LC}^k$ complements the subgraph induced by the sets $=N_{G^k}(i)=\mathcal{X}^k\cap\mathcal{Y}^k + \{j\}$, setting $\mathcal{X}^{k+1} = \mathcal{Y}^{k}$ and $\mathcal{Y}^{k+1} = \mathcal{X}^{k}$. 
\end{itemize}

By repeated application of the above rules, we find:
\begin{equation*}
\begin{split}
\mathcal{X}^6 =     \mathcal{X} \qquad \mathcal{Y}^6 = \mathcal{Y} \qquad \mathcal{Z}^6 = \mathcal{Z} 
\end{split}
\end{equation*}
We have shown the neighbourhoods of $i$ and $j$ have period at most 6. We have yet to show that the edges not involving $i$ or $j$, have undergone one period, which will will do now.

We write the complementation of a graph as an operation on a graph $C : G \rightarrow G^\mathsf{c}$. Further, we denote the complementation of subgraph induced by a set of vertices, $\mathcal{A}$ as $C_\mathcal{A} : G \rightarrow G_{\mathcal{A}^\mathsf{c}}$, where $G_{\mathcal{A}^\mathsf{c}}$ is the input graph $G$ but with the subgraph induced by the vertex set $\mathcal{A}$ complemented. We define $E[A,B]$ as the set of all bipartite edges in $E$ that run from a vertex in the set $A$ to a vertex in the set $B$, $E[\mathcal{A},\mathcal{B}] = \{(a,b)\in E \; : \; a\in\mathcal{A}, \; b\in\mathcal{B}, \; a \neq b \}$. 

Also, we use $C_{E[A , B]} : G \rightarrow G_{E[A,B]^\mathsf{c}}$ do denote bipartite complementation. That is, $G_{E[A,B]^\mathsf{c}}$ is the input graph $G$ but with the bipartite component between $\mathcal{A}$ and $\mathcal{B}$ complemented.

For odd $k$, $\text{LC}^k$ performs the operation $C_{N_{G^k}(j)} = C_{\mathcal{Y}^k\cup\mathcal{Z}^k+\{i\}}$, whilst for even $k$, $\text{LC}^k$ performs the operation $C_{N_{G^k}(i)} = C_{\mathcal{X}^k\cup\mathcal{Y}^k+\{j\}}$. To check the effect of the orbit on the edges of $G^k$ we examine the graph operations in terms of fixed sets of vertices, namely $\mathcal{X}$, $\mathcal{Y}$, $\mathcal{Z}$, using the relations above. Hence the successive \text{LC} operations can be written in the following way:
\begin{equation*}
\text{LC}^1 \;=\;  C_{\mathcal{Y}\cup\mathcal{Z} \cup \{i\}} , \qquad \; \text{LC}^2 \;=\;  C_{\mathcal{X}\cup\mathcal{Z} \cup \{j\}} ,\qquad \ldots\qquad , \qquad \text{etc.}
\end{equation*}
Continuing to apply the rules, we find $\text{LC}^l = \text{LC}^{l+6}$. Hence after 12 successive LCs each complementation $\text{LC}^l$ has cancelled with $\text{LC}^{l+6}$ (since complementation of a fixed set of vertices commutes) and the graph has undergone one period.

Note $C_{\mathcal{A} \cup \mathcal{B}}  = C_{\mathcal{A}}  \;\circ\; C_{\mathcal{B}}  \;\circ\; C_{E[A , B]}$. Using this expansion, the first six operations of the \text{LC} orbit can be written 
\begin{equation*}
\begin{split}
\text{LC}^1 \;=\; & C_{\mathcal{Y}}  \;\circ\;  C_{\mathcal{Z}}  \;\circ\;   C_{E[\mathcal{Y} , \{i\}]}   \;\circ\;  C_{E[\mathcal{Z} , \{i\}]}   \;\circ\;  C_{E[\mathcal{Y} , \mathcal{Z} ]} \\
\text{LC}^2 \;=\; & C_{\mathcal{X}}  \;\circ\;  C_{\mathcal{Z}}  \;\circ\;   C_{E[\mathcal{X} , \{j\}]}   \;\circ\;  C_{E[\mathcal{Z} , \{j\}]}   \;\circ\;  C_{E[\mathcal{X} , \mathcal{Z} ]} \\
\text{LC}^3 \;=\; & C_{\mathcal{X}}  \;\circ\;  C_{\mathcal{Y}}  \;\circ\;   C_{E[\mathcal{X} , \{i\}]}   \;\circ\;  C_{E[\mathcal{Y} , \{i\}]}   \;\circ\;  C_{E[\mathcal{X} , \mathcal{Y} ]} \\
\text{LC}^4 \;=\; & C_{\mathcal{Z}}  \;\circ\;  C_{\mathcal{Y}}  \;\circ\;   C_{E[\mathcal{Z} , \{j\}]}   \;\circ\;  C_{E[\mathcal{Y} , \{j\}]}   \;\circ\;  C_{E[\mathcal{Z} , \mathcal{Y} ]} \\
\text{LC}^5 \;=\; & C_{\mathcal{Z}}  \;\circ\;  C_{\mathcal{X}}  \;\circ\;   C_{E[\mathcal{Z} , \{i\}]}   \;\circ\;  C_{E[\mathcal{X} , \{i\}]}   \;\circ\;  C_{E[\mathcal{Z} , \mathcal{X} ]} \\
\text{LC}^6 \;=\; & C_{\mathcal{Y}}  \;\circ\;  C_{\mathcal{X}}  \;\circ\;   C_{E[\mathcal{Y} , \{j\}]}   \;\circ\;  C_{E[\mathcal{X} , \{j\}]}   \;\circ\;  C_{E[\mathcal{Y} , \mathcal{X} ]} 
\end{split}
\end{equation*}
Noting $C_{E[\mathcal{B} , \mathcal{A} ]}  = C_{E[\mathcal{A} , \mathcal{B} ]}$ (arguments of bipartite edges commute), $[C_\mathcal{A},C_\mathcal{B}] = 0$ (complementations of a fixed set of vertices commute with one another) and $C_\mathcal{A} \circ C_\mathcal{A} =  \mathds{1} $ (complementation is self-inverse), we find $\text{LC}^1 \circ \text{LC}^2 \circ \text{LC}^3 \circ \text{LC}^4 \circ \text{LC}^5 \circ \text{LC}^6 = \mathds{1}$ since each complementation is performed twice. 	Hence \text{LC} operations on only two vertices have period six.

\end{proof}

\clearpage

\subsection{\textsc{FindAccessibleClasses} Algorithm}

\begin{algorithm}[h!]
\DontPrintSemicolon
\SetAlgoLined
\KwData{A resource graph $G(E,V)$ of order $n$

\qquad \: \! \! A function generating graphs of allowed qubit interactions, $t \in T$\;
\qquad \: \! \! The set of sets of entanglement classes of order $n$, $S^n$\;
\qquad \: \! \! Convergence criteria, $d$, the proportion of iterations to be run since the last novel class was found\;
}

\KwResult{Outputs the indices of LU classes of that are accessible with a given resource, $G(E,V)$, using only postselected fusion and $\text{CZ}^{LO}$ gates, as well as a recipes for each.}
$\mathcal{L}_R', H^R ,P^R   \leftarrow  \emptyset$ \tcp*[r]{H will be the recipe for each class, indexed by class} 
$j  \leftarrow 0$                           \tcp*[r]{L' are the classes which can be accessed} 

$G \leftarrow R$\;
\While(\tcp*[f]{Stop when convergence ratio is reached}){$c < d \cdot j$}{         
$p \leftarrow 1$ 									\tcp*[r]{Success probability of this this experiment}     
 $j  \leftarrow  j+1$                \tcp*[r]{Number of iterations done so far } 
 $t(E,V) \leftarrow \textsc{RandMember}(T)$    \tcp*[r]{Choose an ordered gate topology from the allowed set T} 
 $L \leftarrow \textsc{RandInt}(n-1 - |E[G]|, |E[t]|) $              \tcp*[r]{L is the number of gates to perform from the tree} 
 \For{$i \leftarrow$ 1 \KwTo $L$}{
  
  $r \leftarrow  \textsc{RandMember}(E[t])$                  \tcp*[r]{Choose the next edge from the gate topology} 
  $g_r \leftarrow \textsc{RandMember}(\{F_r, \text{CZ}^{LO}_r\})$     \tcp*[r]{we will apply either CZ or fusion along the chosen edge} 

  Apply $g_{r}$ to $G(E,V)$                               \tcp*[r]{Apply the gate} 
  Append ``$g_r$'' to $h$                                 \tcp*[r]{Save  what we have done to the recipe index} 

  \uIf{$g=\text{CZ}^{LO}$}{
    $p \leftarrow p\times (1/9)$                                 \tcp*[r]{Keep track of success probability} 
   $m \leftarrow \textsc{RandInt}(0,5)$                        \tcp*[r]{LC periodic with period 6 for two vertices} 
 \For{$k \leftarrow$ 1 \KwTo $m$}{
  $\alpha \leftarrow r_{(1+k \; (\text{mod} 2))}$         \tcp*[r]{Want to apply LC alernatingly to i, j, i, etc.} 
  $G(E,V) \leftarrow  \text{LC} _{\alpha}(G(E,V))$        \tcp*[r]{Apply LC} 
  Append ``$\text{LC}_\alpha$'' to $h$                    \tcp*[r]{Save what has been done to the recipe} 
}}

  \uIf{$g=F$}{
      $p \leftarrow p\times (1/9)$                                 \tcp*[r]{Keep track of success probability} 
   $m \leftarrow \text{rand}(0,14)$                      \tcp*[r]{the Width of largest LC class is 14}              
 \For{$k \leftarrow$ 1 \KwTo $m$}{                                    
  $\alpha \leftarrow \textsc{RandMember}(V_G [N_G(i) \cup N_G(i) + \{i\} + \{j\}])$     \tcp*[r]{LC does not commute with F} 
  $G(E,V) \leftarrow  \text{LC} _{\alpha}(G(E,V))$                    \tcp*[r]{Hence need to LC whole union of neighbourhoods of i, j} 
  Append ``$\text{LC}_\alpha$'' to $h$                                \tcp*[r]{Save what has been done to the recipe} 
}

	}

    }
  \uIf(\tcp*[f]{Did we find a new class of graph?}){$G(E,V)$ is not equivalent to any graph in classes $\mathcal{L}_R'$}
   {
	Append $i$ to $\mathcal{L}_R'$ where $S^n_i$ is the equivalence class of $G(E,V)$       \tcp*[r]{Save which class we accessed} 
	$H^R_j \leftarrow h$                                                                      \tcp*[r]{Save successful recipe} 
	$P^R_j \leftarrow p$                                                                      \tcp*[r]{Save success probability} 
	$c \leftarrow j$                                                                        \tcp*[r]{Keep track of convergence criteria} 
	}
	
	\uElseIf{$p > P^R_j$ and $G(E,V)$ $\mathcal{L}_{R'_j}$  \text{is equivalent to} $G(E,V)$ $\mathcal{L}_{R'_j}$ }{
	Replace $H^R_j$ with $h$                                                                  \tcp*[r]{Save imrpoved recipe} 
	
	Replace $P^R_j$ with $p$                                                                        \tcp*[r]{Save imrpoved probability} 
}
  } 

Return $\{ \mathcal{L}_R ', H \}$\;

\caption{\textsc{FindAccessibleClasses}}
\end{algorithm}


\begin{figure*}[t!]
\centering
\captionsetup{width=0.90\textwidth}
\includegraphics[width=1.0\textwidth]{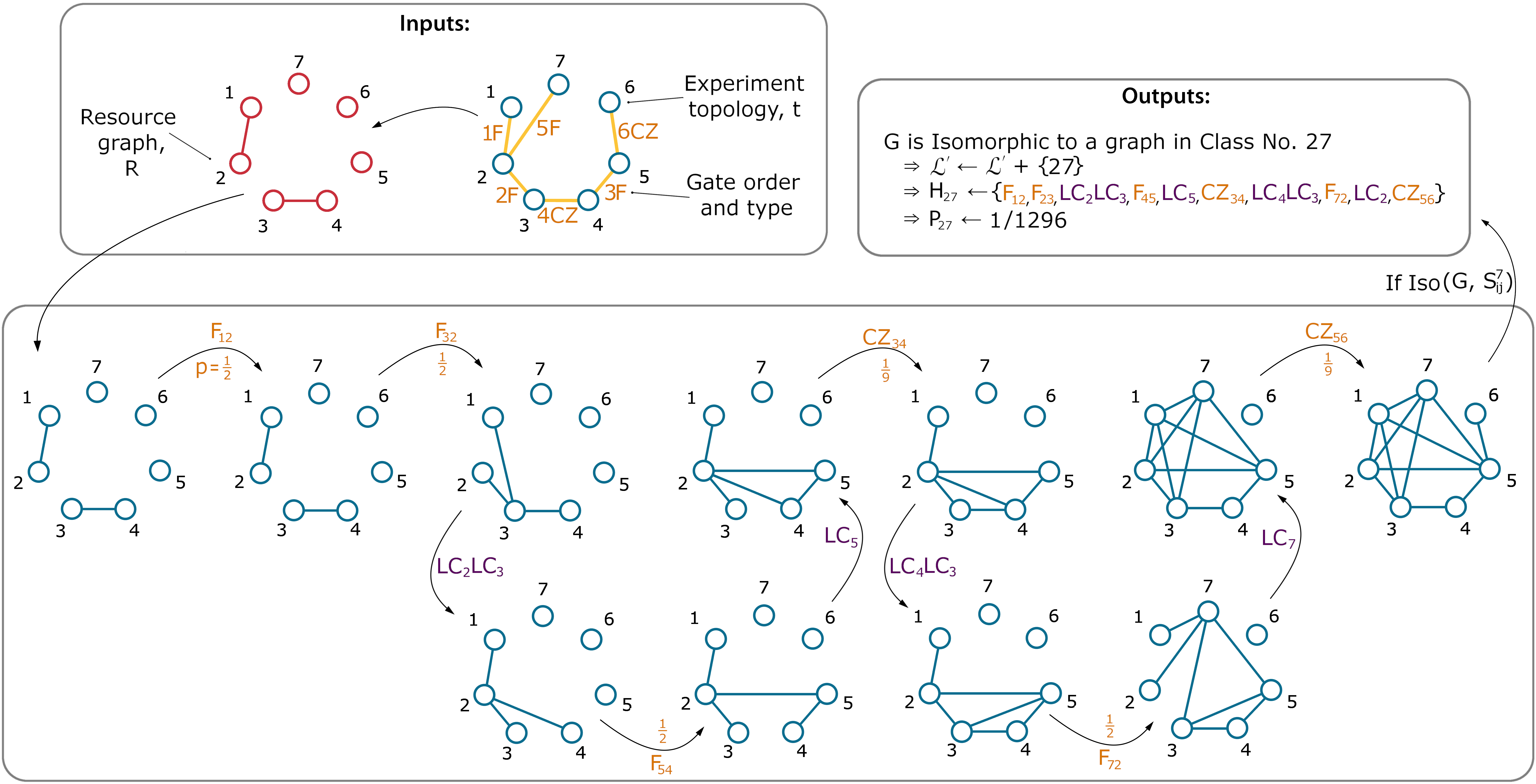}

\caption{An alternate example of one Monte Carlo iteration of \textsc{FindAccessibleClasses}. Starting from a given resource state, in this case 2 EPPs and 3 single photons, operations for experiment topology are performed in order, interspersed with relevant \text{LC}s on the previously acted upon vertices. Pairs of qubits populated with photons from EPPs are highlighted in pink. If the resulting graph is not isomorphic to any graph found thus far, the entanglement class of the graph $i$ is saved to a set $\mathcal{L}_R'$. After many runs, $\mathcal{L}_R'\approx \mathcal{L}_R$.}
\label{alg2}
\end{figure*}

\vspace{3cm}

\clearpage

\subsection{Enumeration of Graph States}

\begin{figure*}[!ht]
\centering
\captionsetup{width=0.90\textwidth}
\includegraphics[width=1.0\textwidth]{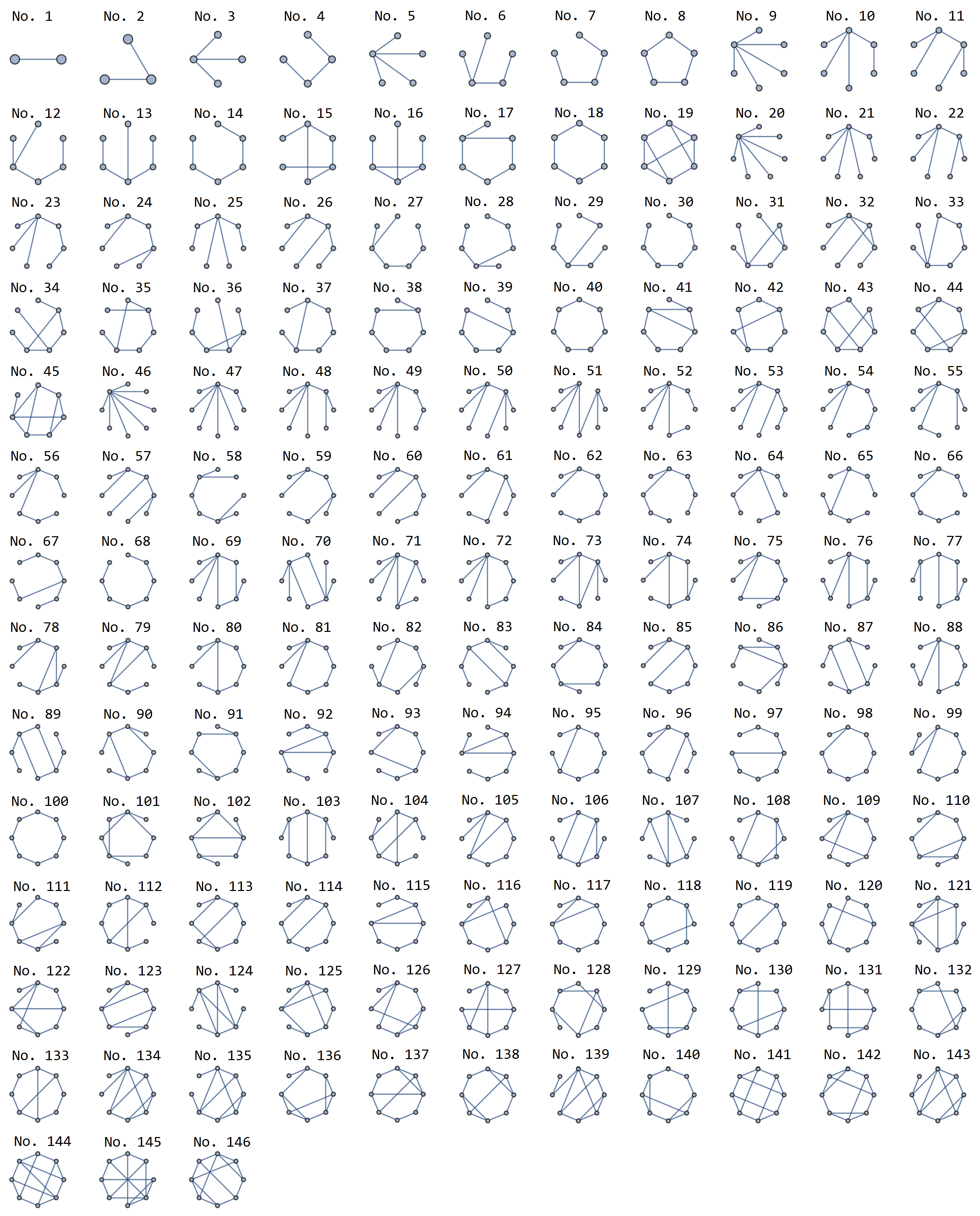}

\caption{Minimal edge count representatives from each of the LU classes up to 8 qubits, canonically numbered as in ref. \cite{cabello2011optimal, hein2004multiparty, danielsen2006classification}.}
\label{bigtable}
\end{figure*}

\subsection{Example of postselected gate - \text{CZ}}
The postselected $\text{CZ}^{LO}$ acting on $\ket{\text{++}}$ produces the following state:
\begin{equation*}
\begin{split}
\text{CZ}^{LO}\ket{++}= &\frac{1}{3}\ket{101000}_f + \frac{1}{3\sqrt{2}}\ket{100100}_f + \frac{1}{3\sqrt{2}}\ket{100010}_f \\ 
+ & \frac{1}{3} \ket{100001}_f +  \frac{1}{3\sqrt{2}} \ket{011000}_f +\quad  \;  \frac{1}{6}\ket{010100}_f   \\
+ & \frac{1}{6}\ket{010010}_f  + \frac{1}{3\sqrt{2}}\ket{010001}_f  +\quad  \;  \frac{1}{6}\ket{001100}_f \\
- & \frac{1}{6}\ket{001010}_f   -\frac{1}{3\sqrt{2}}\ket{001001}_f  -\frac{1}{3\sqrt{2}}\ket{000110}_f\\
+ & \frac{1}{3} \ket{000101}_f -\quad  \;  \; \frac{1}{3} \ket{002000}_f  +\quad  \; \frac{1}{3}\ket{000200}_f
\end{split}
\end{equation*}
Where the non-qubit terms in
\begin{equation*}
\begin{split}
\mathcal{J}  =  \text{span}( &\ket{101000}_f,\ket{100100}_f,  \ket{100010}_f, \\ 
& \ket{100001}_f,\ket{011000}_f, \ket{010001}_f,\\ 
& \ket{001001}_f,\ket{000110}_f, \ket{000101}_f,\\ 
&\qquad \: \: \: \: \qquad \ket{002000}_f,\ket{000200}_f )
\end{split}
\end{equation*} 
are removed by postselection $P_\mathcal{Q}$. Hence
\begin{equation*}
\begin{split}
P_{\mathcal{Q}} \text{CZ}^{LO} \ket{++} = & \frac{1}{6}\ket{010100}_f +\frac{1}{6}\ket{010010}_f \\
                                        + & \frac{1}{6}\ket{001100}_f -\frac{1}{6}\ket{001010}_f \\
                                        = & \frac{1}{6}\ket{00} + \frac{1}{6}\ket{01} + \frac{1}{6}\ket{10} - \frac{1}{6}\ket{11}  \\
                                        = & \frac{1}{3} \text{CZ} \ket{++}
\end{split}
\end{equation*}

\subsection{Postselection of multiple entangled pairs from squeezed vacuum}

Pairwise entangled states of $n$ pairs of photons are commonly generated by two postselecting the $n$-photon subspace of $\frac{n}{2}$ coherently pumped EPP source. Unfortunately, these states contain junk states affect postselectability. This is because it is not possible to distinguish the case where there were one pair of photons is generated in each source, and the case where when some sources produce more than one pair, which is at least as likely (for $n$ photons in total). To see this, take two (unnormalised) fock states produced by a EPP source.

\begin{equation*}
\begin{split}
\ket{\xi}^{\otimes2} = (\ket{0000}  +\gamma&\ket{1010}  + \; \; \gamma\ket{0101} + \\
\gamma^2 & \ket{2020} + \gamma^2 \ket{0202}  + \dots)^{\otimes2}
\end{split}
\end{equation*}
Where $\ket{\Phi^+} = \ket{1010}_f +\ket{0101}_f$ for two pairs of two modes comprising two qubits. The $O(\gamma^2)$, four-photon terms:
\begin{equation*}
\begin{split}
\ket{\xi}^{\otimes2} = &\gamma^2\ket{\Phi^+ \Phi^+} + \gamma^2 \ket{11110000}  + \gamma^2 \ket{20200000}  \\
+ &\gamma^2 \ket{02020000}  + \gamma^2 \ket{00001111} + \gamma^2 \ket{00002020} \\
+ &\gamma^2 \ket{00000202}
\end{split}
\end{equation*}
The postselected state, $\ket{\Phi^+ \Phi^+}$ makes up only a minority of the four photon state, and two-photon-per-qubit terms dominate. Similarly, for larger ensembles of sources, each permutation of pairs being produced the sources is present in the superposition, and must be considered. For example, in the six photon subspace of three sources, there are terms where all three pairs were produced in just one source, as well as terms where just one source produced an extra pair. Postselected sources produce mostly junk states, and how these traverse and experiment must be considered when evaluating whether an experiment will successfully postselect.
\clearpage

\end{document}